\documentclass{llncs}
\usepackage{graphicx}
\usepackage{hyperref}
\usepackage{xcolor}
\usepackage{xspace}
\usepackage{paralist}
\usepackage{amsmath,amssymb}
\usepackage{subfig}
\usepackage{multirow}
\usepackage{dsfont}
\usepackage{wrapfig}
\usepackage{todonotes}
\usepackage{enumitem}
\usepackage{ amssymb }

\newif\ifarxiv
\arxivtrue

\newtheorem{claimx}{Claim}

\hypersetup{pdfborder={0 0 0.4}}

\renewenvironment{proof}
{{\bf Proof:}}{\hspace*{\fill}\fbox{}\par\vspace{2mm}}

\graphicspath{{img/}}

\newcommand{\nt}{{NT}\xspace}

\newcommand{\NP}{$\mathbb{NP}$\xspace}

\renewcommand{\NP}{$\mathbb{NP}$\xspace}

\newcommand{\NPC}{\mbox{$\mathbb{NP}$-complete}\xspace}

\newcommand{\NPCN}{\mbox{$\mathbb{NP}$-completeness}\xspace}

\newcommand{\NPHN}{\mbox{$\mathbb{NP}$-hardness}\xspace}

\definecolor{blue}{rgb}{0.274,0.392,0.666}
\definecolor{red}{rgb}{0.627,0.117,0.156}
\definecolor{Red}{rgb}{1,0,0}
\definecolor{green}{rgb}{0,0.588,0.509}
\definecolor{gray}{rgb}{0.5,0.5,0.5}

\renewenvironment{proof}
{\noindent{\em Proof:}}{\hspace*{\fill}\qed\par\vspace{2mm}}
\newcommand{\remove}[1]{}

\newcommand{\try}{Computing NodeTrix Representations of\\ Clustered Graphs}
\title{\try\thanks{Research partially supported by MIUR project AMANDA, prot. 2012C4E3KT\_001.}}

\author{Giordano {Da Lozzo}, Giuseppe {Di Battista}, Fabrizio {Frati}, Maurizio {Patrignani}
\institute{Roma Tre University, Italy\\
\email{\{dalozzo,gdb,frati,patrigna\}@dia.uniroma3.it}
  }
}

\begin{document}
\ifarxiv
\pagestyle{plain}
\else
\fi
\maketitle

\begin{abstract}
NodeTrix representations are a popular way to visualize clustered graphs; they represent clusters as adjacency matrices and inter-cluster edges as curves connecting the matrix boundaries. We study the complexity of constructing NodeTrix representations focusing on planarity testing problems, and we show several \NPCN results and some polynomial-time algorithms. Building on such algorithms we develop a JavaScript library for NodeTrix representations aimed at reducing the crossings between edges incident to the same matrix.
\end{abstract}


\section{Introduction and Overview} \label{se:intro}

NodeTrix representations have been introduced by Henry, Fekete, and McGuffin~\cite{hfm-nhvsn-07} in one of the most cited papers of the InfoVis conference~\cite{citevis}. 
A NodeTrix representation is a hybrid representation for the visualization of social networks where the node-link paradigm is used to visualize the overall structure of the network, within which adjacency matrices show communities.

Formally, a NodeTrix (\nt for short) representation is defined as follows. 
A {\em flat clustered graph} $(V,E,\mathcal{C})$ is a graph $(V,E)$ with a partition $\mathcal{C}$ of $V$ into sets $V_1,\dots,V_k$, called {\em clusters},  that can be defined according to the application needs. The word ``flat'' is used to underline that clusters are not arranged in a multi-level hierarchy (see, e.g.,~\cite{cdfk-atcefcg-14,df-ectefcgsf-j-09} for two papers dealing with non-flat clustered graphs). An edge $(u,v) \in E$ with $u \in V_i$ and $v \in V_j$ is an {\em intra-cluster edge} if $i=j$ and is an {\em inter-cluster edge} if $i\neq j$. 
In an {\em \nt representation} clusters $V_1,\dots,V_k$ are represented by non-overlapping symmetric adjacency matrices $M_1,\dots,M_k$, where $M_i$ is drawn in the plane
so that its boundary is a square $Q_i$ with sides parallel to the
coordinate axes. Thus, the matrices $M_1,\dots,M_k$ convey the information
about the intra-cluster edges of $(V,E,{\cal C})$, while each inter-cluster
edge $(u,v)$ with $u \in V_i$ and $v \in V_j$ is represented by a curve
connecting a point on $Q_i$ with a point on $Q_j$, where the point on $Q_i$ (on
$Q_j$) belongs to the column or to the row of $M_i$ (resp.\ of $M_j$) associated with $u$ (resp.\ with $v$). 

Several papers aimed at improving the readability of \nt representations by reducing the number of crossings between inter-cluster edges.
For this purpose, vertices can have duplicates in different matrices~\cite{hbf-ircsn-08} or clusters can be computed so to have dense intra-cluster graphs and a planar inter-cluster graph~\cite{bbdlpp-lbars-11}. 

In this paper we study the problem of automatically constructing an \nt representation of a given flat clustered graph. This problem combines traditional graph drawing issues, like the placement of a set of geometric objects in the plane (here the squares $Q_1,\dots,Q_k$) and the routing of the graph edges (here the inter-cluster edges), with a novel algorithmic challenge: To handle the degrees of freedom given by the choice of the \emph{order} for the rows and the columns of the matrices and by the choice of the \emph{side} of the matrices to which the inter-cluster edges attach to. Indeed, the order of the rows and columns of a matrix $M_i$ is arbitrary, as long as $M_i$ is symmetric; further, an inter-cluster edge incident to $M_i$ can arbitrarily exit $M_i$ from four sides: left or right if it exits $M_i$ from its associated row, or top or bottom if it exits $M_i$ from its associated column. 

When working on a new model for graph representations, the very first step is usually to study the complexity of testing if a graph admits a planar representation within that model. Hence, in Section~\ref{se:complexity} we deal with the problem of testing if a flat clustered graph admits a planar \nt representation. An \nt representation is {\em planar} if no inter-cluster edge $e$ intersects any matrix $M_i$, except possibly at an end-point of $e$ on $Q_i$, and no two inter-cluster edges $e$ and $e'$ cross each other, except possibly at a common end-point. The {\sc Nodetrix Planarity} ({\sc \nt Planarity} for short) problem asks if a flat clustered graph admits a  planar \nt representation.

Our findings show how tough the problem is (see Table~\ref{tab:complexity}).
Namely, we show that {\sc \nt Planarity} is \NPC and remains so 
even if the
order of the rows and of the columns of the matrices is fixed (i.e., it is part of the
input), or if the exit sides of the inter-cluster edges on the matrix boundaries are fixed. It is easy to show that {\sc \nt Planarity} becomes linear-time solvable if both the row-column order and the exit sides of the inter-cluster edges are fixed. But this is probably too restrictive for practical applications since all the degrees of freedom that are representation-specific are lost.

\begin{table}[tb!]
\centering
\begin{tabular}{|p{2.2cm}|p{1cm}|p{2cm}|p{2cm}||p{2cm}|p{2cm}|}

\cline{3-6}
\multicolumn{2}{c}{\multirow{2}{*}{}} & \multicolumn{2}{|c||}{{\bf General Model}} &
\multicolumn{2}{|c|}{{\bf Monotone Model}} \\
\cline{3-6}
\multicolumn{2}{c|}{} & {\bf Free Sides} & {\bf Fixed Sides} & {\bf Free Sides} & {\bf Fixed Sides} \\
\hline
\multirow{2}{*}{\begin{minipage}{2cm}\bf Row/Column Order\end{minipage}} & {\bf Free} &
\textcolor{red}{$\mathds{NPC}$}~[Th.~\ref{th:free-ordering-free-sides}] &
\textcolor{red}{$\mathds{NPC}$}~[Th.~\ref{th:free-ordering-fixed-sides}] &
\textcolor{red}{$\mathds{NPC}$}~[Th.~\ref{th:monotone-free-ordering-free-sides}]
&
\textcolor{red}{$\mathds{NPC}$}~[Th.~\ref{th:monotone-free-ordering-fixed-sides}
]\\
\cline{2-6}
& {\bf Fixed} &
\textcolor{red}{$\mathds{NPC}$}~[Th.~\ref{th:fixed-ordering-free-sides}] &
\textcolor{green}{$\mathds{P}$}~[Th.~\ref{th:fixed-ordering-fixed-sides}] &
\textcolor{green}{$\mathds{P}$}~[Th.~\ref{th:monotone-fixed-ordering-free-sides}
]$\dag$ &
\textcolor{green}{$\mathds{P}$}~[Th.~\ref{th:monotone-fixed-order-and-side-polynomial}]\\
\hline
\cline{3-6}
\cline{3-6}
\end{tabular}
\medskip
\caption{Complexity results for {\sc \nt Planarity}. The result marked $\dag$ assumes that the number of clusters is constant.}
\label{tab:complexity}
\end{table}

Motivated by such complexity results, in Section~\ref{se:monotone} we study a
more constrained model
that is 
still useful for practical applications.
A {\em monotone \nt representation} is an \nt representation in which
the matrices have prescribed positions and the inter-cluster edges are represented by $xy$-monotone curves inside the convex hull of their incident matrices. We require that this convex hull, which might contain many edges, 
does not intersect any other matrix.
We study this model for two reasons.  First, in most of (although not in all) the
available examples of \nt representations the inter-cluster edges are represented by $xy$-monotone curves (see, e.g., the NodeTrix clips and prototype available online~\cite{nodetrix-demo}). Second, we are interested in supporting a visualization system where the position of the matrices is decided by the user and the inter-cluster edges are automatically drawn with ``few'' crossings. Therefore, the crossings between inter-cluster edges not incident to a common matrix are somehow unavoidable, as they depend on the matrix positions selected by the users, and we are only interested in reducing the number of \emph{local} crossings, that are the crossings between pairs of edges incident to the same matrix.

We say that an \nt representation is {\em locally planar} if no two inter-cluster edges incident to the same matrix cross. While testing if a flat clustered graph admits a monotone \nt locally planar representation is \NPC even if the sides are fixed (see Table~\ref{tab:complexity}), the problem becomes polynomial-time solvable in the reasonable scenario in which the number of matrices is constant, the order of the rows and columns is fixed, and the sides of the matrices to which the inter-cluster edges attach is variable.

Building on the insights for the last result, we developed a library (Section~\ref{se:editor}) for \nt representations (a demo is available online~\cite{giordano-demo}). The adopted techniques allow the user to move the matrices around while the layout of the inter-cluster edges is automatically recomputed; this happens without any slowdown of the interaction.

Conclusions and open problems are discussed in Section~\ref{se:conclusions}
where {\sc \nt Planarity} is related to graph drawing problems of theoretical
interest.

Before proceeding to prove our results, we establish formal definitions and notation. An {\nt representation} consists of:\begin{enumerate}
\item A {\em row-column order} $\sigma_i$ for each cluster $V_i$, that is, a bijection $\sigma_i: V_i \leftrightarrow \{1,\dots,|V_i|\}$.
\item A {\em side assignment} $s_i$ for each inter-cluster edge incident to $V_i$, that is, an injective mapping $s_i: \bigcup_{j\neq i}{E_{i,j}} \rightarrow \{\textrm{\sc{t}}, \textrm{\sc{b}}, \textrm{\sc{l}}, \textrm{\sc{r}}\}$, where $E_{i,j}$ is the set of inter-cluster edges between the clusters $V_i$ and $V_j$ ($V_i$ and $V_j$ are \emph{adjacent} if $E_{i,j}\neq \emptyset$).
%
\item A {\em matrix} $M_i$ for each cluster $V_i$, that is, a representation of $V_i$ as a symmetric adjacency matrix such that: 
\begin{enumerate}
\item the boundary of $M_i$ is a square $Q_i$ with sides parallel to the coordinate axes; let $\min_x(Q_i)$ be the minimum $x$-coordinate of a point on $Q_i$; $\min_y(Q_i)$, $\max_x(Q_i)$, and $\max_y(Q_i)$ are defined analogously; 
\item the left-to-right order of the columns and the top-to-bottom order of the rows in $M_i$ is $\sigma_i$; and 
\item any two distinct matrices are disjoint; if $V_i$ has only one vertex, we often talk about the matrix representing that vertex, rather than the matrix representing $V_i$.
\end{enumerate} 
\item An {\em edge drawing} for each inter-cluster edge $e=(u,v)$ with $u \in V_i$ and $v \in V_j$, that is, a representation of $e$ as a Jordan curve between two points $p_u$ and $p_v$ defined as follows. Let $m^u_{\textrm{\sc{t}}}$ be the mid-point of the line segment that is the intersection of the top side of $Q_i$ with the column associated to $u$ in $M_i$; points $m^u_{\textrm{\sc{b}}}$, $m^u_{\textrm{\sc{l}}}$, and $m^u_{\textrm{\sc{r}}}$ are defined analogously. Then $p_u$ coincides with $m^u_{\textrm{\sc{t}}}$, $m^u_{\textrm{\sc{b}}}$, $m^u_{\textrm{\sc{l}}}$, or $m^u_{\textrm{\sc{r}}}$ if $s_i(e)=\textrm{\sc{t}}$, $s_i(e)=\textrm{\sc{b}}$, $s_i(e)=\textrm{\sc{l}}$, or $s_i(e)=\textrm{\sc{r}}$, respectively. Point $p_v$ is defined analogously. 
\end{enumerate}

\section{Testing NodeTrix Planarity}\label{se:complexity}

In this section we study the time complexity of testing {\sc NodeTrix Planarity} for a flat clustered graph. We start with the following.

\begin{theorem}\label{th:free-ordering-free-sides}
{\sc NodeTrix Planarity} is \NPC even if at most three clusters contain more than one vertex.
\end{theorem}


\begin{proof}
The membership in $\mathbb{NP}$ of {\sc \nt Planarity} is proved in Lemma~\ref{le:np}. 

For the $\mathbb{NP}$-hardness we give a reduction from an $\mathbb{NP}$-complete problem called {\sc Partitioned $3$-Page Book Embedding}~\cite{adn-aspbep-15}, whose input is a graph $(V, E)$ with $E$ partitioned into three sets $E_1$, $E_2$, and $E_3$. The problem asks whether a total ordering $\cal O$ of $V$ exists such that the end-vertices of any two edges $e$ and $e'$ in the same set $E_i$ do not alternate in $\cal O$; $e$ and $e'$ {\em alternate} if an end-vertex of $e'$ is between the two end-vertices of $e$ in $\cal O$ and vice versa. 

  \begin{figure}[htb]
    \centering
    \subfloat[]{
      \includegraphics[height=0.24\textwidth,page=3]{img/hardness-variableorder}
      \label{fi:lemma-variableorder-a}
    }
    \subfloat[]{
      \includegraphics[height=0.29\textwidth,page=1]{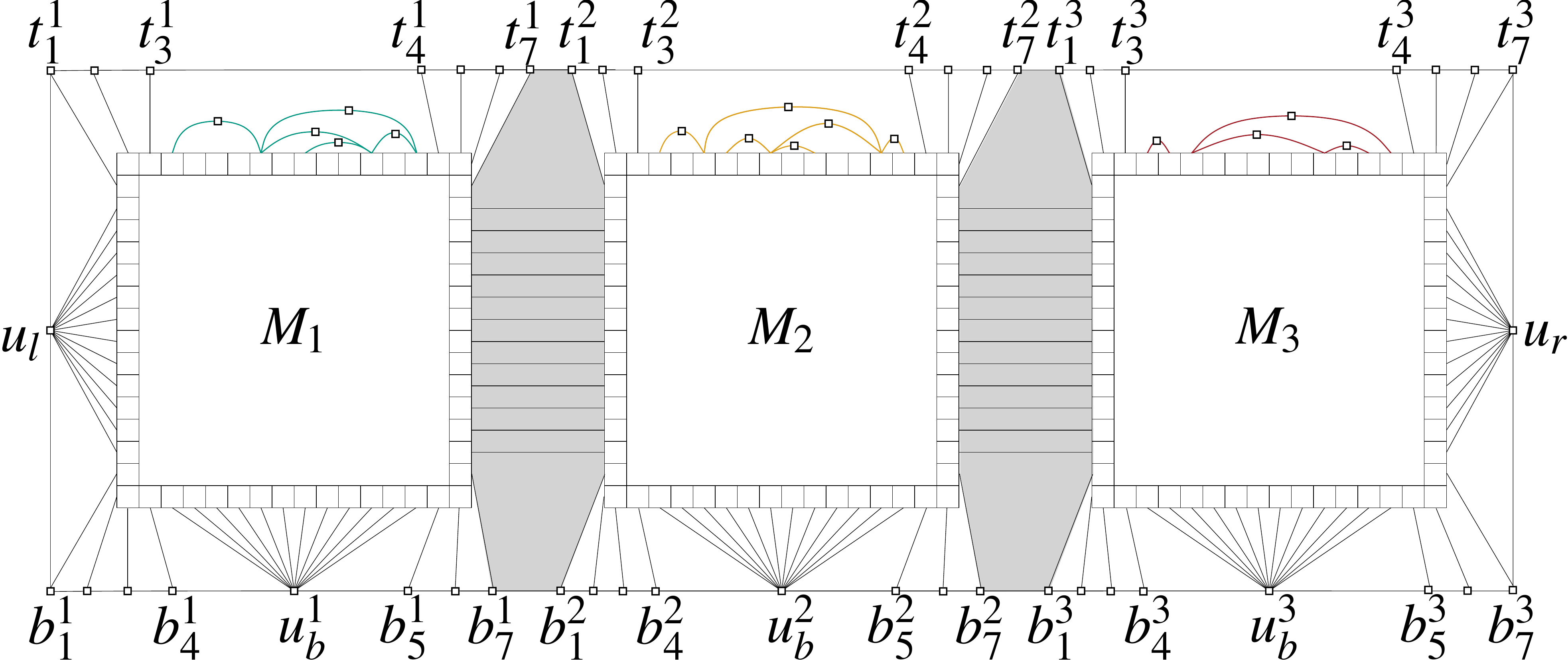}
      \label{fi:hardness-variableorder-b}
    }
    \caption{
      (a) An instance $(V, E= E_1\cup E_2 \cup E_3)$ of {\sc Partitioned $3$-Page Book Embedding} and (b) the corresponding instance $(V',E', \mathcal{C}')$ of {\sc \nt Planarity}. The gray regions are $R^{1,2}$ and $R^{2,3}$.}
    \label{fi:hardness-variableorder}
  \end{figure}

We show how to construct in polynomial time an instance $(V',E',\mathcal{C}')$ of {\sc \nt Planarity} starting from an instance $(V, E = E_1\cup E_2\cup E_3)$ of {\sc Partitioned $3$-Page Book Embedding}. Refer to Fig.~\ref{fi:hardness-variableorder}. We define $\mathcal{C}'$ (and hence implicitly $V'$) as follows. 

\begin{itemize}
\item Set $\mathcal{C}'$ contains three clusters $V''_i=V'_i \cup \{x_i,y_i,w_i,z_i\}$ with $i=1,2,3$, where sets $V'_i$ are in bijection with $V$ (and hence with each other); we denote by $v'_i$ the vertex in $V'_i$ that is in bijection with a vertex $v\in V$; 
\item for every edge $e\in E$, set $\mathcal{C}'$ contains a cluster $\{u'_e\}$; and
\item set $\mathcal{C}'$ contains clusters $\{u^1_b\}$, $\{u^2_b\}$, $\{u^3_b\}$, $\{u_l\}$, $\{u_r\}$ and, for $i=1,\dots,7$, clusters $\{t^1_i\}$, $\{t^2_i\}$, $\{t^3_i\}$, $\{b^1_i\}$, $\{b^2_i\}$, and $\{b^3_i\}$. 
\end{itemize}

The set $E'$ contains an arbitrary set of intra-cluster edges and the following inter-cluster edges.

\begin{itemize}
\item {\em equivalence edges}: edges $(u'_e,r'_i)$ and $(u'_e,s'_i)$, for every edge $e=(r,s)\in E_i$;
\item {\em bounding edges}: the edges of cycle $D=(u_l,t^1_1,\dots,t^1_7,t^2_1,\dots,t^2_7,t^3_1,\dots,t^3_7,u_r,\\b^3_7,b^3_6,b^3_5,u^3_b,b^3_4,\dots,b^3_1,b^2_7,b^2_6,b^2_5,u^2_b,b^2_4,\dots,b^2_1,b^1_7,b^1_6,b^1_5,u^1_b,b^1_4,\dots,b^1_1,u_l)$;
\item {\em order-preserving edges}: edges $(v'_1,v'_2)$ and $(v'_2,v'_3)$, for every vertex $v\in V$;
\item {\em side-filling edges}: edges between $u_l$ and every vertex in $V'_1$, edges between $u_r$ and every vertex in $V'_3$, and edges between $u^i_b$ and every vertex in $V'_i$ for $i=1,2,3$; and
\item {\em corner edges}: for $i=1,2,3$, edges $(t^i_1,y_i)$, $(t^i_2,x_i)$, $(t^i_3,y_i)$, $(t^i_4,w_i)$, $(t^i_5,z_i)$, $(t^i_6,x_i)$, $(t^i_7,y_i)$, $(b^i_1,w_i)$, $(b^i_2,z_i)$, $(b^i_3,x_i)$, $(b^i_4,y_i)$, $(b^i_5,w_i)$, $(b^i_6,z_i)$, and $(b^i_7,w_i)$.
\end{itemize}

The described construction can easily be performed in polynomial time. We now prove the equivalence between the instance $(V, E)$ of {\sc Partitioned $3$-Page Book Embedding} and the instance $(V',E',\mathcal{C}')$ of {\sc \nt Planarity}. 

The direction ($\Longrightarrow$) is easy to prove. Suppose that $(V, E)$ is a positive instance of {\sc Partitioned $3$-Page Book Embedding} and let $\cal O$ be a total ordering of $V$ such that the end-vertices of any two edges $e$ and $e'$ in the same set $E_i$ do not alternate in $\cal O$. For $i=1,2,3$, let $\sigma_i$ be the total order of the vertices in $V''_i$ such that $x_i$, $y_i$, $w_i$, and $z_i$ are the first, second, last but one, and last vertex in $\sigma_i$, respectively, and the vertices in $V'_i$ (which all follow $x_i$ and $y_i$ and precede $w_i$ and $z_i$ in $\sigma_i$) are ordered so that, for every $r'_i,s'_i\in V'_i$, vertex $r'_i$ precedes vertex $s'_i$ in $\sigma_i$ if and only if $r$ precedes $s$ in $\cal O$. For $i=1,2,3$, we represent $V''_i$ as a symmetric adjacency matrix $M_i$ whose left-to-right order of the columns is $\sigma_i$; every other cluster in $\mathcal{C}'$ consists of a single vertex and we arbitrarily define a matrix for it. We embed these matrices in the plane (except for the matrices representing vertices $u'_e$ with $e\in E$, which will be embedded later) so that no two of them overlap. The side assignment for the inter-cluster edges is as follows.

First, we assign: 
\begin{itemize}
\item edges $(t^i_j,t^i_{j+1})$, $(t^i_7,t^{i+1}_1)$, $(b^i_j,b^i_{j+1})$, $(b^i_4,u^i_b)$, $(u^i_b,b^i_5)$, and $(b^i_7,b^{i+1}_1)$ to the right (left) side of the matrix representing the first (resp.\ second) vertex in the pair; 
\item edges $(b^1_1,u_l)$, $(u_l,t^1_1)$, $(b^3_7,u_r)$, and $(u_r,t^3_7)$ to the top (bottom) side of the matrix representing the first (resp.\ second) vertex in the pair;
\item every other edge incident to $t^i_j$ to the bottom side of the matrix representing $t^i_j$; 
\item every other edge incident to $b^i_j$ to the top side of the matrix representing $b^i_j$; 
\item every other edge incident to $u^i_b$ to the top side of the matrix representing $u^i_b$; 
\item every other edge incident to $u_l$ to the right side of the matrix representing $u_l$; and 
\item every other edge incident to $u_r$ to the left side of the matrix representing $u_r$.
\end{itemize}

By suitably routing the bounding edges, we get that $M_1$, $M_2$, and $M_3$ we ensure that the cycle $D$ bounds the outer face of the representation, with $M_1$, $M_2$, and $M_3$ inside it. 

Second, we assign the side-filling and order-preserving edges so that none of them is assigned to the top side of a matrix $M_i$. To achieve this, we assign: the side-filling edges incident to $u_l$ and to $u_r$ to the left side of $M_1$ and to the right side of $M_3$, respectively; the side-filling edges incident to $u^i_b$ to the bottom side of $M_i$; and the order-preserving edges between vertices in $V'_i$ and vertices in $V'_{i+1}$ to the right side of $M_i$ and to the left side of $M_{i+1}$. Route all these edges inside $D$. In particular, the order-preserving edges can be routed without crossings since the top-to-bottom order of the vertices in $V'_i$ along the right side of $M_i$ is $\sigma_i-\{x_i,y_i,w_i,z_i\}$, since the top-to-bottom order of the vertices in $V'_{i+1}$ along the left side of $M_{i+1}$ is $\sigma_{i+1}-\{x_{i+1},y_{i+1},w_{i+1},z_{i+1}\}$, and since $\sigma_i-\{x_i,y_i,w_i,z_i\}$ and $\sigma_{i+1}-\{x_{i+1},y_{i+1},w_{i+1},z_{i+1}\}$ are the same ordering according to the bijection between $V'_i$ and $V'_{i+1}$. 

Third, we assign the corner edges $(t^i_2,x_i)$, $(t^i_3,y_i)$, $(t^i_4,w_i)$, and $(t^i_5,z_i)$ to the top side of $M_i$, $(t^i_6,x_i)$, $(t^i_7,y_i)$, and $(b^i_7,w_i)$ to the right side of $M_i$, $(b^i_3,x_i)$, $(b^i_4,y_i)$, $(b^i_5,w_i)$, and $(b^i_6,z_i)$ to the bottom side of $M_i$, and $(t^i_1,y_i)$, $(b^i_1,w_i)$, and $(b^i_2,z_i)$ to the left side of $M_i$. The corner edges can routed inside $D$ without crossing the side-filling and order-preserving edges since by construction their end-vertices in $V''_i$ are the first, second, last but one, and last vertex in $\sigma_i$.

Fourth, embed the matrices representing $u'_e$ with $e\in E_i$ in the region $R_i$ delimited by edges $(t^3_i,t^4_i)$, $(t^3_i,y_i)$, and $(t^4_i,w_i)$, and by the top side of $M_i$. Assign the equivalence edges $(u'_e,r'_i)$ and $(u'_e,s'_i)$ to the top side of $M_i$ and to any sides of the matrix representing $u'_e$; route these edges in $R_i$. This can be done without introducing crossings, since the left-to-right order $\sigma_i-\{x_i,y_i,w_i,z_i\}$ of the vertices in $V'_i$ along the top side of $M_i$ is the same as the order $\cal O$ of the vertices in $V$ (according to the bijection between $V'_i$ and $V$), hence no two pairs of edges $((p'_i,u'_e),(u'_e,q'_i))$ and $((r'_i,u'_f),(u'_f,s'_i))$ have alternating end-points along the top side of $M_i$ given that edges $(p,q)$ and $(r,s)$ do not have alternating end-vertices in $\cal O$.

The proof of the direction ($\Longleftarrow$) is more involved. Given a solution $\Gamma$ for the instance $(V',E',\mathcal{C}')$ of {\sc \nt Planarity}, we need to define an ordering $\cal O$ for the vertex set $V$ in the instance $(V, E)$ of {\sc Partitioned $3$-Page Book Embedding}. In order to do that, we prove some claims and lemmata about the structure of $\Gamma$. For $i=1,2,3$, let $M_i$ be the matrix representing $V''_i$ in $\Gamma$, let $Q_i$ be its boundary, and let $\sigma_i$ be the left-to-right order of the vertices in $V''_i$ along the top side of $Q_i$. We have the following.

\begin{claimx} \label{cl:matrices-arrangement-D}
The matrices $M_1$, $M_2$, and $M_3$, and the matrices representing vertices $u'_e$ for all $e\in E$ are  on the same side of $D$ in $\Gamma$. 
\end{claimx}

\begin{proof}
The statement follows from the fact that the clusters $V''_i$ for $i=1,2,3$ and $\{u'_e\}$ for $e\in E$ are connected by inter-cluster edges that are not bounding edges. Indeed, the clusters $V''_1$, $V''_2$, and $V''_3$ are connected by order-preserving edges and each cluster $\{u'_e\}$ is connected to a cluster $V''_i$ by two equivalence edges. Thus, if the matrices representing two of these clusters were on opposite sides of $D$ in $\Gamma$, there would exist: (i) an order-preserving or an equivalence edge crossing a matrix representing a vertex in $D$ or crossing a bounding edge, or (ii) a matrix among $M_1$, $M_2$, and $M_3$ or a matrix representing a vertex $u'_e$ with $e\in E$ overlapping a matrix representing a vertex in $D$ or crossing a bounding edge. However, this would
 contradict the planarity of $\Gamma$.
\end{proof}

We henceforth assume that $M_1$, $M_2$, and $M_3$, as well as the matrices representing the vertices $u'_e$ for all $e\in E$, are {\em inside} $D$ in $\Gamma$. Indeed, by Claim~\ref{cl:matrices-arrangement-D}, these matrices are on the same side of $D$ in $\Gamma$. If they are outside $D$, then the matrices representing vertices in $D$ are all incident to an internal face $f$ of $\Gamma$. Hence, changing the outer face of $\Gamma$ to $f$ ensures that $M_1$, $M_2$, and $M_3$ and the matrices representing the vertices $u'_e$ for all $e\in E$ are inside $D$ in $\Gamma$. This change of the outer face can be accomplished by rerouting the inter-cluster edges that are dual to a simple path in the dual of $\Gamma$ from $f$ to the outer face.

In the following we also assume that $t^1_1$, $t^1_2$, and $t^1_3$ are encountered in this order when traversing $D$ in clockwise direction. This is not a loss of generality, up to a reflection of~$\Gamma$.

Let $\Gamma''$ be the restriction of $\Gamma$ to $M_1$, $M_2$, and $M_3$, to the matrices representing the vertices $t^i_j$, $b^i_j$, $u^i_b$, $u_l$, and $u_r$, and to the bounding and corner edges of $(V',E',\mathcal{C}')$. For $i=1,2$, let $R^{i,i+1}$ be the region of the plane delimited by the edges $(t^i_7,y_i)$, $(b^i_7,w_i)$, $(t^i_7,t^{i+1}_1)$, $(b^i_7,b^{i+1}_1)$, $(t^{i+1}_1,y_{i+1})$, and $(b^{i+1}_1,w_{i+1})$, by the boundaries of $M_i$ and $M_{i+1}$, and by the boundaries of the matrices representing vertices $t^i_7$, $b^i_7$, $t^{i+1}_1$, and $b^{i+1}_1$. We have the following.

\begin{claimx} \label{cl:matrices-disjoint}
Every order-preserving edge connecting a vertex in $V'_i$ with a vertex in $V'_{i+1}$ lies inside $R^{i,i+1}$. 
\end{claimx}

\begin{proof}
Assume that $i=1$; the discussion with $i=2$ is analogous. 

We define three regions that partition the interior of $\Gamma''$. Region $R^1$ is the minimal simple (i.e., without holes) region of the plane containing $M_1$ and containing the matrices and edges representing path $D^1=(w_1,b^1_7,b^1_6,b^1_5,u^1_b,b^1_4,b^1_3,b^1_2,$ $b^1_1,u_l,t^1_1,\dots,t^1_7,y_1)$. Region $R^2$ is the minimal simple region of the plane containing $M_2$, $M_3$, and containing the matrices and edges representing path $D^2=(y_2,t^2_1,\dots,t^2_7,t^3_1,\dots,t^3_7,u_r,b^3_7,b^3_6,b^3_5,u^3_b,b^3_4,b^3_3,b^3_2,b^3_1,b^2_7,b^2_6,b^2_5,u^2_b,b^2_4,b^2_3,b^2_2,b^2_1,w_2)$. The third region is $R^{i,i+1}$. 

Since $D^1$ and $M_1$ do not share vertices with $D^2$, $M_2$, and $M_3$, since $D^1$ and $D^2$ are both incident to the outer face of $\Gamma''$, and since $R^1$ and $R^2$ are simple, we have that $R^1$ and $R^2$ are disjoint. Further, the only faces of $\Gamma''$ regions $R^1$ and $R^2$ are both incident to (and which the order-preserving edges lie because of the planarity of $\Gamma$) are the outer face of $\Gamma''$ and $R^{i,i+1}$; however, neither $M_i$ nor $M_{i+1}$ is incident to the outer face of $\Gamma''$, since $M_i$ and $M_{i+1}$ lie inside $D$. It follows that every order-preserving edge connecting a vertex in $V'_i$ with a vertex in $V'_{i+1}$ lies inside $R^{i,i+1}$. 
\end{proof}


We now present the following claim, which argues about the incidences between the corner edges and the square $Q_i$.

  \begin{figure}[tb!]
    \centering
    \subfloat[]{
      \includegraphics[height=0.28\textwidth,page=1]{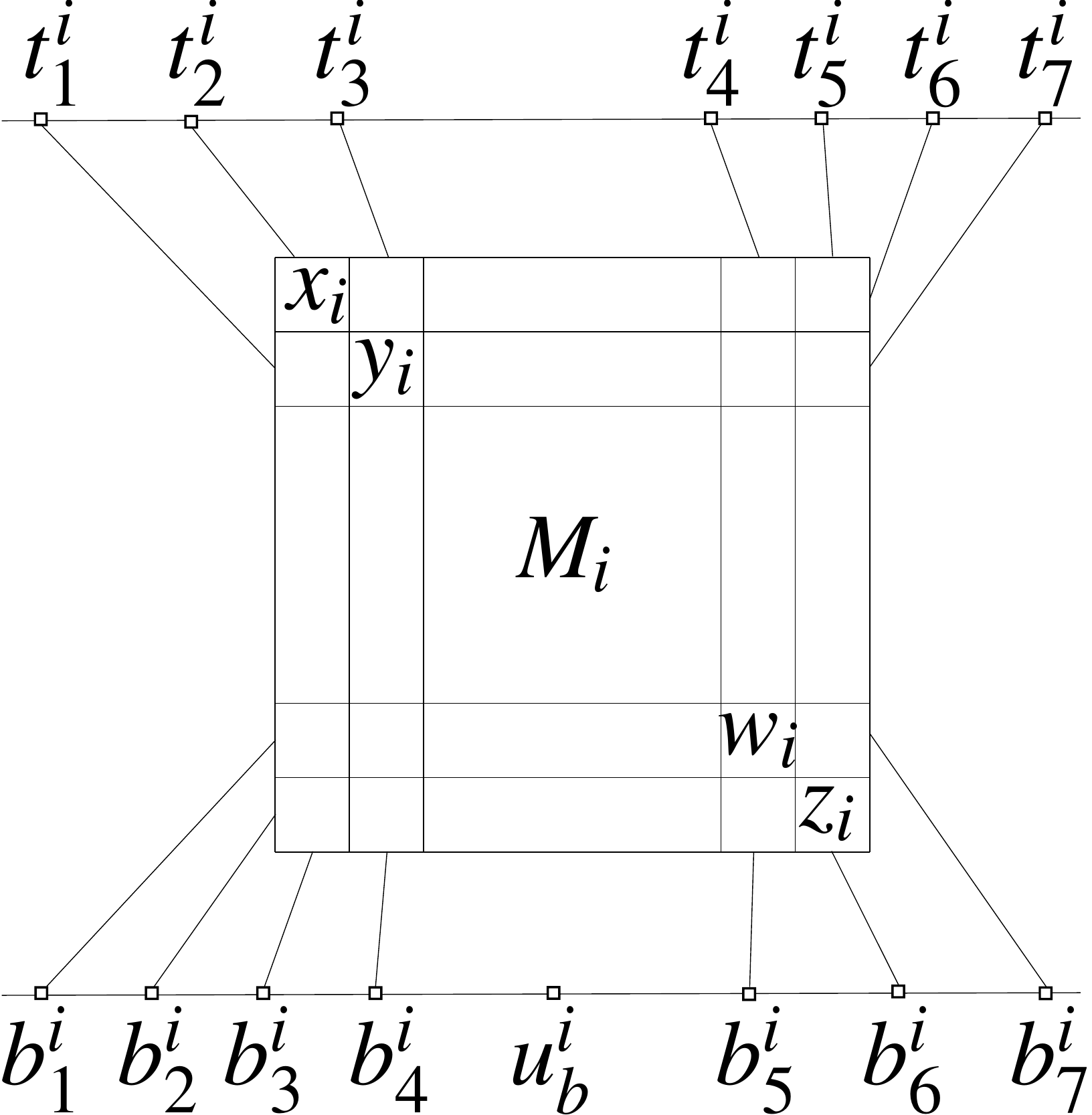}
      \label{fi:corners-sub1}
    }
    \hspace{8mm}
    \subfloat[]{
      \includegraphics[height=0.28\textwidth,page=1]{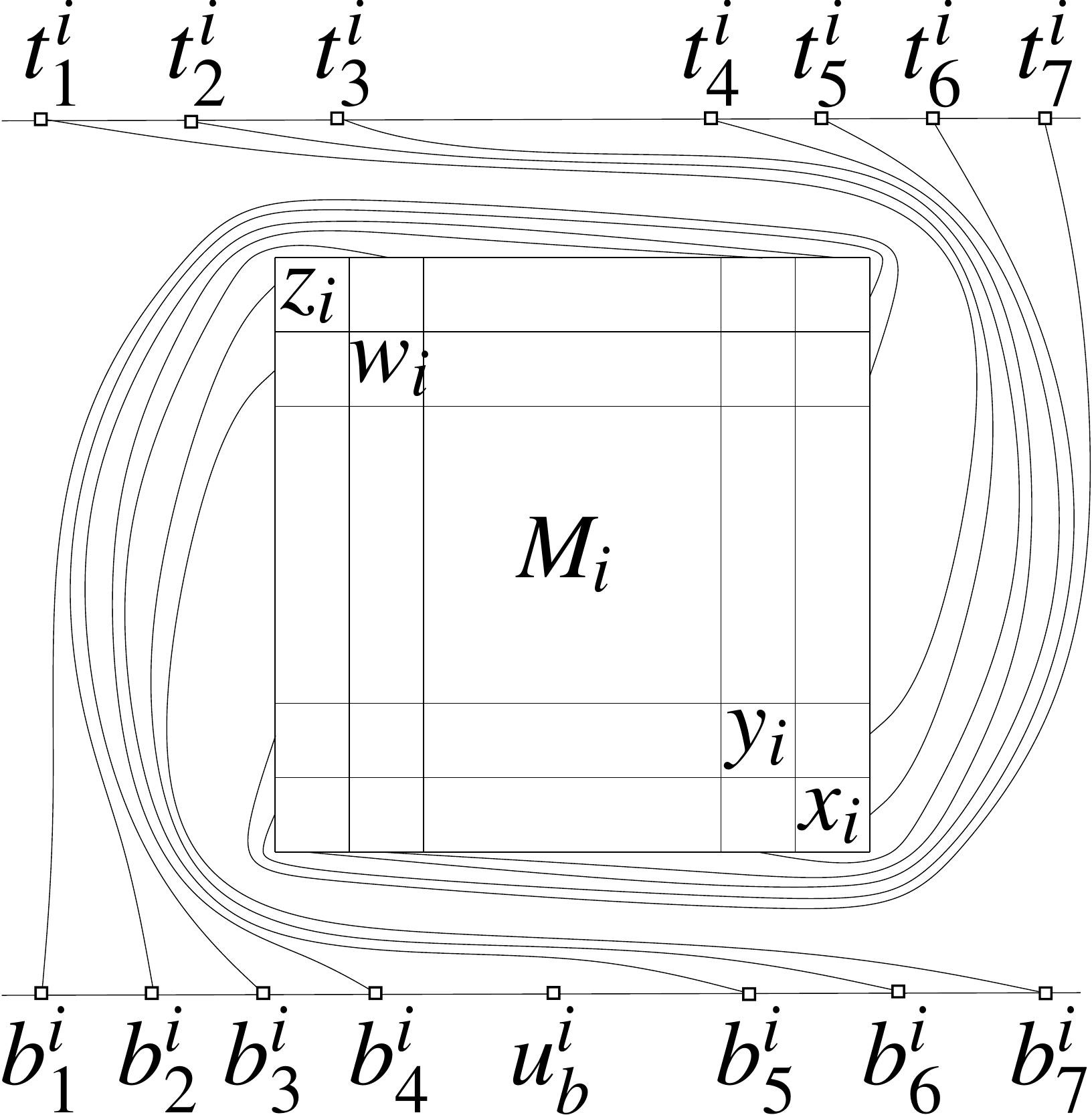}
      \label{fi:corners-sub2}
    }
    \caption{(a) Statement (1) of Claim~\ref{cl:corners}. (b) Statement (2) of Claim~\ref{cl:corners}.}
    \label{fi:corners}
  \end{figure}

\begin{claimx} \label{cl:corners}  
For each $i\in \{1,2,3\}$, one of the following statements holds true (see respectively Figs.~\ref{fi:corners-sub1} and~\ref{fi:corners-sub2}).

\begin{enumerate}
\item[(1)] Vertices $x_i$, $y_i$, $w_i$, and $z_i$ are the first, second, last but one, and last vertex in $\sigma_i$, respectively. Further, $(t^i_3,y_i)$ and $(t^i_4,w_i)$ are assigned to the top side of $Q_i$, $(t^i_7,y_i)$ and $(b^i_7,w_i)$ are assigned to the right side of $Q_i$, $(b^i_4,y_i)$ and $(b^i_5,w_i)$ are assigned to the bottom side of $Q_i$, and $(t^i_1,y_i)$ and $(b^i_1,w_i)$ are assigned to the left side of $Q_i$.
\item[(2)] Vertices $z_i$, $w_i$, $y_i$, and $x_i$ are the first, second, last but one, and last vertex in $\sigma_i$, respectively. Further, $(b^i_5,w_i)$ and $(b^i_4,y_i)$ are assigned to the top side of $Q_i$, $(b^i_1,w_i)$ and $(t^i_1,y_i)$ are assigned to the right side of $Q_i$, $(t^i_4,w_i)$ and $(t^i_3,y_i)$ are assigned to the bottom side of $Q_i$, and $(b^i_7,w_i)$ and $(t^i_7,y_i)$ are assigned to the left side of $Q_i$.
\end{enumerate}
\end{claimx}   

\begin{proof}
We prove that, if $y_i$ precedes $w_i$ in $\sigma_i$, then statement (1) holds true. A similar proof shows that, if $w_i$ precedes $y_i$ in $\sigma_i$, then statement (2) holds true.

First, each of the four corner edges $(t^i_1,y_i)$, $(t^i_3,y_i)$, $(t^i_7,y_i)$, and $(b^i_4,y_i)$ incident to $y_i$ is assigned to a distinct side of $Q_i$ in $\Gamma$. Indeed assume, for a contradiction, that two of these corner edges, say $(t^i_1,y_i)$ and $(t^i_3,y_i)$, are assigned to the same side of $Q_i$, as in Fig.~\ref{fi:claim3-sub1}; this implies that the end-points of $(t^i_1,y_i)$ and $(t^i_3,y_i)$ on $Q_i$ coincide. Let $p^*$ be the end-point of these edges on $Q_i$; thus, $p^*$ is on the boundary of a row or column of $M_i$ associated to $y_i$. Let $R^*$ be the region delimited by the corner edges $(t^i_1,y_i)$ and $(t^i_3,y_i)$, by the bounding edges $(t^i_1,t^i_2)$ and $(t^i_2,t^i_3)$, and by the boundaries of the matrices representing $t^i_1$, $t^i_2$, and $t^i_3$. Suppose that $Q_i$ is not contained in $R^*$. If edge $(t^i_2,x_i)$ leaves the matrix representing $t^i_2$ outside $D$, then it crosses $D$ since $M_i$ is inside $D$, a contradiction. Otherwise, $(t^i_2,x_i)$ leaves the matrix representing $t^i_2$ inside $R^*$. Since $p^*$ is the only point on $Q_i$ incident to $R^*$ and since $p^*$ is not on the boundary of a row or column  of $M_i$ associated to $x_i$, we have that $(t^i_2,x_i)$ crosses the boundary of $R^*$, a contradiction. If $Q_i$ is contained in $R^*$, a contradiction can be derived analogously by considering the routing of edge $(t^i_6,x_i)$ rather than $(t^i_2,x_i)$.

  \begin{figure}[htb]
    \centering
    \subfloat[]{
      \includegraphics[width=0.2\textwidth]{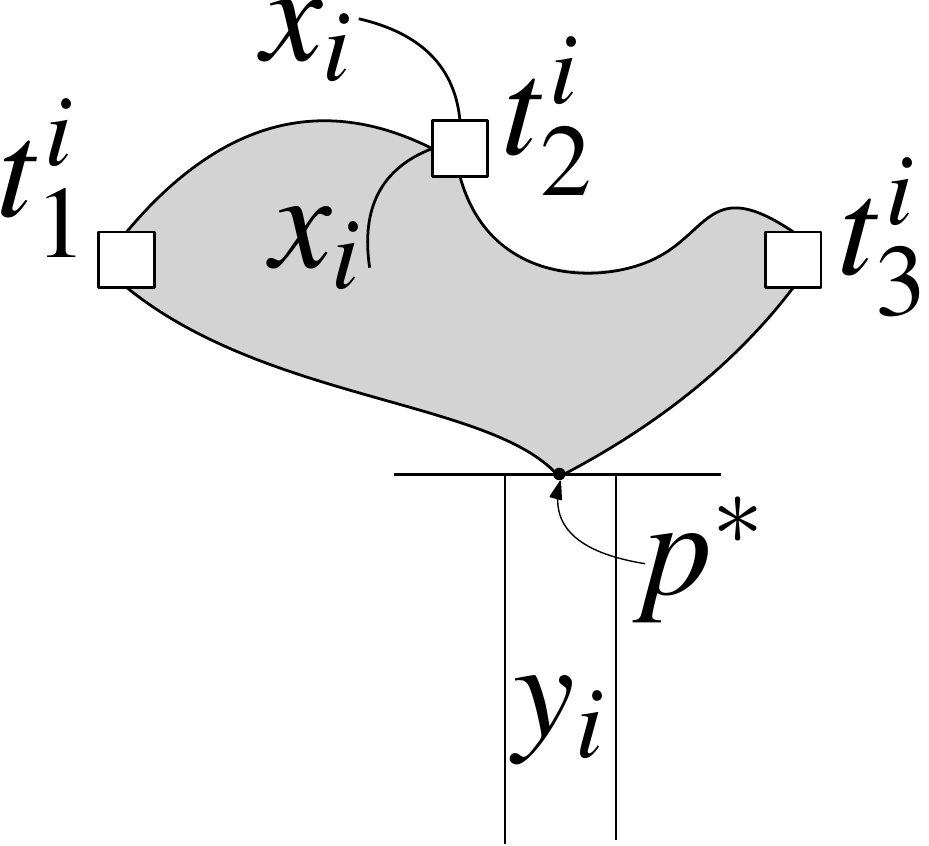}
      \label{fi:claim3-sub1}
    }
    \hspace{0.1\textwidth}
    \subfloat[]{
      \includegraphics[width=0.25\textwidth]{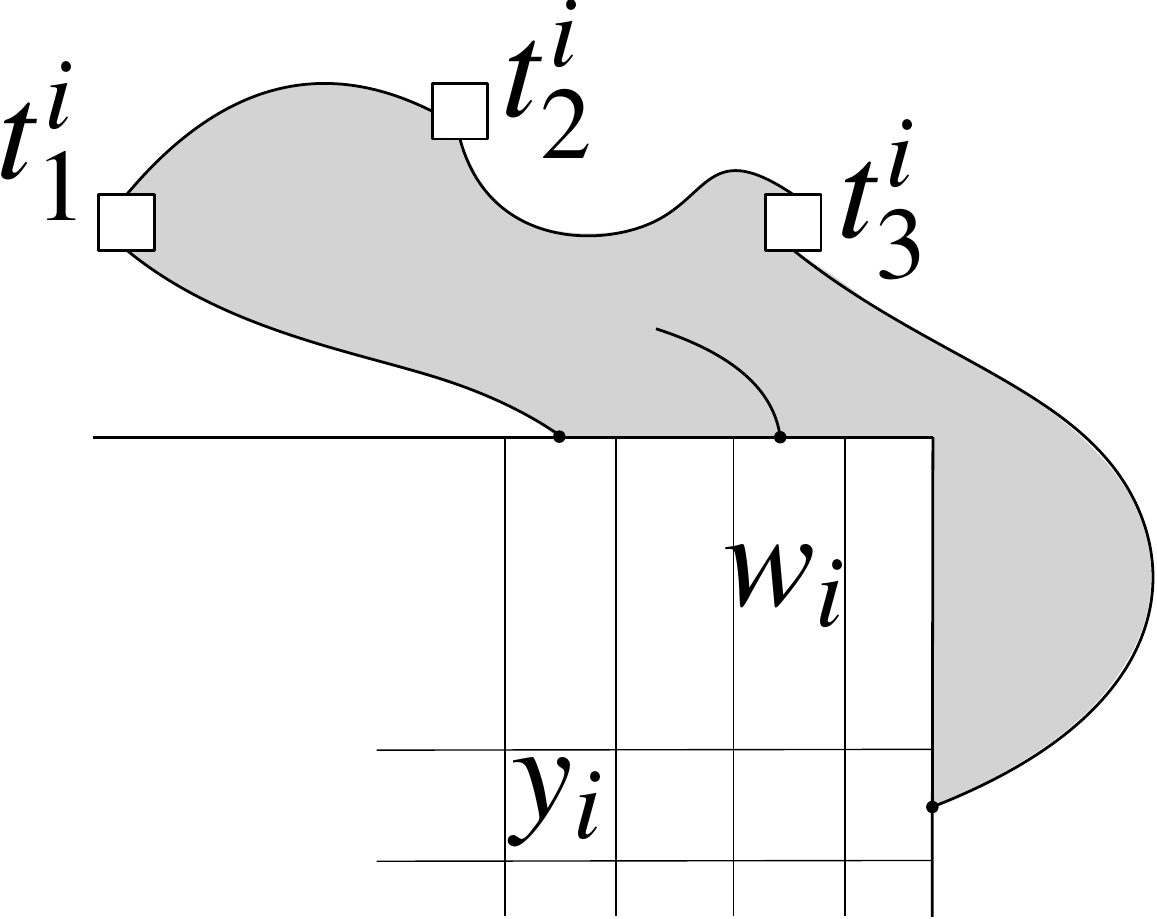}
      \label{fi:claim3-sub2}
    }
    \hspace{0.1\textwidth}
    \subfloat[]{
      \includegraphics[width=0.22\textwidth]{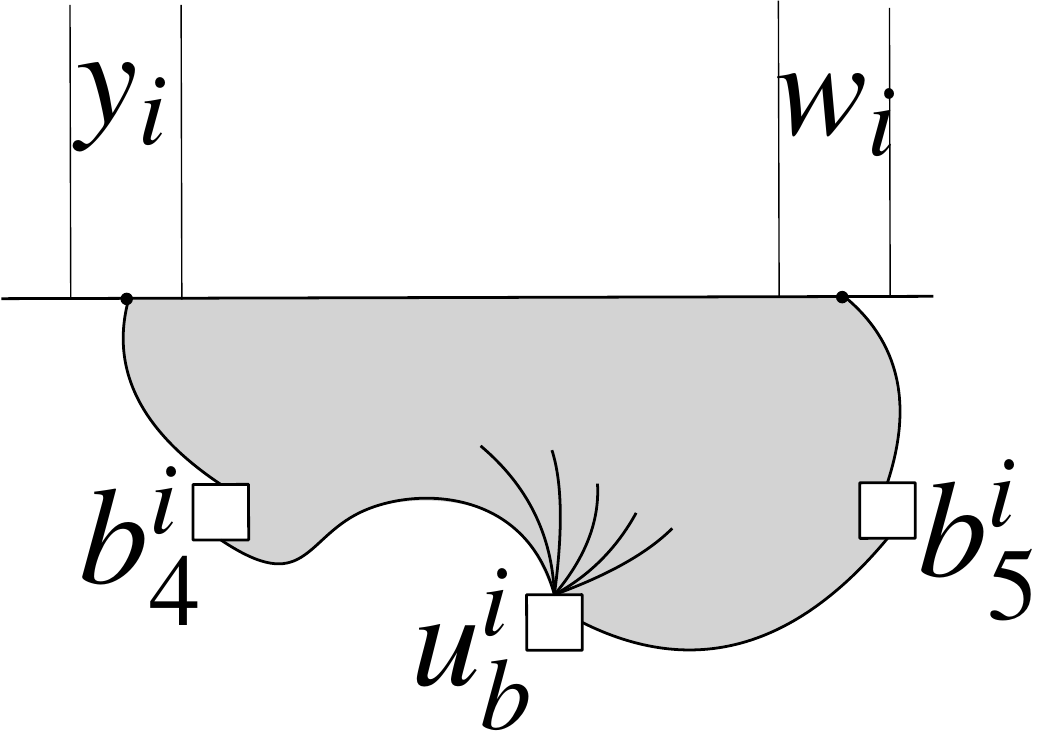}
      \label{fi:claim3-sub3}
    }
    \caption{(a) Contradiction to the planarity of $\Gamma$ if $(t^i_1,y_i)$ and $(t^i_3,y_i)$ are assigned to the same side of $Q_i$. The gray region is $R^*$. (b) Contradiction to the planarity of $\Gamma$ if $(t^i_1,y_i)$ is assigned to the top side of $Q_i$. (c) All the vertices in $V'_i$ come after $y_i$ and before $w_i$ in $\sigma_i$. The gray region is $R_b$.}
    \label{fi:claim3}
  \end{figure}

A similar argument proves that each of the four corner edges $(t^i_4,w_i)$, $(b^i_7,w_i)$, $(b^i_5,w_i)$, and $(b^i_1,w_i)$ incident to $w_i$ is assigned to a distinct side of $Q_i$ in $\Gamma$.

We now prove that $(t^i_1,y_i)$ is assigned to the left side of $Q_i$. Assume the contrary, for a contradiction; refer to Fig.~\ref{fi:claim3-sub2}, where $(t^i_1,y_i)$ is assigned to the top side of $Q_i$. Consider the line segment $Q^+$ traversed when walking along $Q_i$ in clockwise direction from the end-point of $(t^i_1,y_i)$ to the end-point of $(t^i_3,y_i)$. If $(t^i_1,y_i)$ is assigned to the top, right, or bottom side of $Q_i$, we have that $Q^+$ entirely contains the top side of the column of $M_i$ associated to $w_i$, the right side of the row of $M_i$ associated to $w_i$, or the left side of the row of $M_i$ associated to $w_i$, respectively. Since one of the four corner edges incident to $w_i$ has its end-point on this segment and since none of $t^i_1$, $t^i_2$, and $t^i_3$ is adjacent to $w_i$, it follows that: (i) one of the corner edges $(t^i_1,y_i)$ and $(t^i_3,y_i)$, or one of the corner edges incident to $w_i$ crosses a matrix representing a vertex in $D$ or crosses one of the bounding edges; or (ii) two among the corner edges $(t^i_1,y_i)$ and $(t^i_3,y_i)$, and among the corner edges incident to $w_i$ cross each other. In both cases we get a contradiction to the planarity of $\Gamma$.  

Since $(t^i_1,y_i)$ is assigned to the left side of $Q_i$, since each of the four corner edges incident to $y_1$ is assigned to a distinct side of $Q_i$, and since $t^i_1$, $t^i_3$, $t^i_7$, and $b^i_4$ appear in this clockwise order along $D$, we have that $(t^i_3,y_i)$, $(t^i_7,y_i)$, and $(b^i_4,y_i)$ are assigned to the top, right, and bottom side of $Q_i$, respectively.

An analogous proof shows that $(t^i_4,w_i)$, $(b^i_7,w_i)$, $(b^i_5,w_i)$, and $(b^i_1,w_i)$ are assigned to the top, right, bottom, and left side of $Q_i$, respectively. 

Further, $(t^i_2,x_i)$ is assigned to the left or top side of $Q_i$, since $(t^i_1,y_i)$ and $(t^i_3,y_i)$ are assigned to the left and top side of $Q_i$, respectively, and since $t^i_1$, $t^i_2$, and $t^i_3$ appear in this clockwise order along $D$. In both cases, $x_i$ precedes $y_i$ in $\sigma_i$. An analogous argument proves that $z_i$ follows $w_i$ in $\sigma_i$. Hence, $x_i$, $y_i$, $w_i$, and $z_i$ appear in this order in $\sigma_i$.


It remains to argue that $x_i$, $y_i$, $w_i$, and $z_i$ are the first, second, last but one, and last vertex in $\sigma_i$, respectively; refer to Fig.~\ref{fi:claim3-sub3}. Let $R_b$ be the region delimited by the corner edges $(b^i_4,y_i)$ and $(b^i_5,w_i)$, by the bounding edges $(b^i_4,u^i_b)$ and $(u^i_b,b^i_5)$, and by the boundaries of the matrices representing $b^i_4$, $u^i_b$, $b^i_5$, and $V''_i$. If any side-filling edge incident to $u^i_b$ leaves the matrix representing $u^i_b$ outside $D$, then this edge crosses $D$, a contradiction to the planarity of $\Gamma$. Otherwise, every side-filling edge incident to $u^i_b$ leaves the matrix representing $u^i_b$ inside $R_b$. Since $(b^i_4,y_i)$ and $(b^i_5,w_i)$ are both assigned to the bottom side of $Q_i$, it follows that all the side-filling edges incident to $u^i_b$ have their other end-point on the bottom side of $Q_i$, between the end-point of $(b^i_4,y_i)$ and the end-point of $(b^i_5,w_i)$. Thus, all the vertices in $V'_i$ come after $y_i$ and before $w_i$ in $\sigma_i$; this concludes the proof of statement (1). 
\end{proof}


We are now ready to state and prove the following two lemmata.

\begin{lemma} \label{le:same-side}
For each $i\in\{1,2,3\}$, all the equivalence edges belonging to the set $E'_i=\{(u'_e,r'_i),(u'_e,s'_i) : e=(r,s)\in E_i\}$ are assigned to the same side of $Q_i$, except possibly for the edges $(u'_e,r'_i)$ and $(u'_e,s'_i)$ such that $r'_i$ and $s'_i$ are consecutive in $\sigma_i$. 
\end{lemma}

\begin{proof}
We prove that every edge in $E'_i=\{(u'_e,r'_i),(u'_e,s'_i) : e=(r,s)\in E_i\}$ is assigned to the same side of $Q_i$ edges $(t^i_3,y_1)$ and $(t^i_4,w_1)$ are assigned to, except possibly for the edges $(u'_e,r'_i)$ and $(u'_e,s'_i)$ such that $r'_i$ and $s'_i$ are consecutive in $\sigma_i$. The statement clearly implies Lemma~\ref{le:same-side}. Note that $(t^i_3,y_i)$ and $(t^i_4,w_i)$ are assigned either both to the top or both to the bottom side of $Q_i$, by Claim~\ref{cl:corners}. 

We first prove that, for any edge $e=(r,s)\in E_i$, the edges $(u'_e,r'_i)$ and $(u'_e,s'_i)$ are assigned to the same side of $Q_i$. Indeed suppose, for a contradiction, that the edge $(u'_e,r'_i)$ is assigned to, say, the top side of $Q_i$ and the edge $(u'_e,s'_i)$ is assigned to a different side of $Q_i$. By Claim~\ref{cl:corners}, there exist two corner edges that are assigned to the top side of $Q_i$ and that are incident to the second and to the last but one vertex in $\sigma_i$, respectively; these corner edges are $(t^i_3,y_i)$ and $(t^i_4,w_i)$, or $(b^i_5,w_i)$ and $(b^i_4,y_i)$. Assume that they are $(t^i_3,y_i)$ and $(t^i_4,w_i)$, as the other case is analogous. Consider the region $R_t$ delimited by $(t^i_3,y_i)$ and $(t^i_4,w_i)$, by the bounding edge $(t^i_3,t^i_4)$, by the boundaries of the matrices representing $t^i_3$ and $t^i_4$, and by the top side of $Q_i$. By Claim~\ref{cl:corners}, the incidence point of edge $(u'_e,r'_i)$ on the top side of $Q_i$ is between the incidence points of $(t^i_3,y_i)$ and $(t^i_4,w_i)$ on the top side of $Q_i$, given that $r'_i\in V'_i$. Since edge $(u'_e,r'_i)$ does not cross $M_i$, it leaves $M_i$ inside $R_t$. Since $(u'_e,s'_i)$ is assigned to a side of $Q_i$ different from the top side, then edge $(u'_e,r'_i)$, or edge $(u'_e,s'_i)$, or the matrix representing $u'_e$ crosses the boundary of $R_t$, a contradiction to the planarity of $\Gamma$. It follows that the edges $(u'_e,r'_i)$ and $(u'_e,s'_i)$ are assigned to the same side of $Q_i$. 
Now suppose that, for some $e=(r,s)\in E_i$, the edges $(u'_e,r'_i)$ and $(u'_e,s'_i)$ are both assigned to a side different from the one the edges $(t^i_3,y_i)$ and $(t^i_4,w_i)$ are assigned to, and suppose that $r'_i$ and $s'_i$ are not consecutive in $\sigma_i$. By Claim~\ref{cl:corners}, the edges $(u'_e,r'_i)$ and $(u'_e,s'_i)$ are both assigned to the side $(b^i_4,y_i)$ and $(b^i_5,w_i)$ are assigned to, or to the side $(t^i_1,y_i)$ and $(b^i_1,w_i)$ are assigned to, or to the side $(t^i_7,y_i)$ and $(b^i_7,w_i)$ are assigned to.

Assume first that $(u'_e,r'_i)$ and $(u'_e,s'_i)$ are both assigned to the side $(b^i_4,y_i)$ and $(b^i_5,w_i)$ are assigned to. As in the proof of Claim~\ref{cl:corners}, we can define $R_b$ as the region delimited by the edges $(b^i_4,y_i)$, $(b^i_5,w_i)$, $(b^i_4,u^i_b)$, and $(u^i_b,b^i_5)$, and by the boundaries of the matrices representing $b^i_4$, $u^i_b$, $b^i_5$, and $V''_i$; then every side-filling edge incident to $u^i_b$ leaves the matrix representing $u^i_b$ inside $R_b$, as otherwise it would cross $D$. It follows that all the side-filling edges incident to $u^i_b$ have their other end-point on the same side of $Q_i$ edges $(b^i_4,y_i)$ and $(b^i_5,w_i)$ are assigned to. Let $p'_i\in V'_i$ be any vertex  between $r'_i$ and $s'_i$ in $\sigma_i$. This vertex exists by hypothesis. Since $p'_i$ is between $r'_i$ and $s'_i$ in $\sigma_i$, the incidence point of the side-filling edge $(u^i_b,p'_i)$ on $Q_i$ is between the incidence points of $(u'_e,r'_i)$ and $(u'_e,s'_i)$ on $Q_i$. It follows that $(u^i_b,p'_i)$ crosses $M_i$, or the matrix representing $u'_e$, or $(u'_e,r'_i)$, or $(u'_e,s'_i)$.    

The cases in which $(u'_e,r'_i)$ and $(u'_e,s'_i)$ are both assigned to the side $(t^i_1,y_i)$ and $(b^i_1,w_i)$ are assigned to (with $i=1$), or to the side $(t^i_7,y_i)$ and $(b^i_7,w_i)$ are assigned to (with $i=3$) are analogous to the previous one, with $u_l$ or $u_r$ playing the role of $u^i_b$, respectively. 

Assume next that $(u'_e,r'_i)$ and $(u'_e,s'_i)$ are both assigned to the side $(t^i_7,y_i)$ and $(b^i_7,w_i)$ are assigned to and $i\leq 2$. By Claim~\ref{cl:corners} and by the planarity of $\Gamma$, the matrix representing $u'_e$ and the edges $(u'_e,r'_i)$ and $(u'_e,s'_i)$ lie inside region $R^{i,i+1}$. By Claim~\ref{cl:matrices-disjoint}, every order-preserving edge between a vertex in $V'_i$ and a vertex in $V'_{i+1}$ also lies inside $R^{i,i+1}$, hence it is assigned to the same side of $Q_i$ as $(t^i_7,y_i)$ and $(b^i_7,w_i)$. Again by Claim~\ref{cl:corners}, every order-preserving edge between a vertex in $V'_i$ and a vertex in $V'_{i+1}$ has its end-point between the end-points of $(t^i_7,y_i)$ and $(b^i_7,w_i)$ along the side of $Q_i$ these edges are all assigned to. Let $p'_i\in V'_i$ be any vertex  between $r'_i$ and $s'_i$ in $\sigma_i$. This vertex exists by hypothesis. Since $p'_i$ is between $r'_i$ and $s'_i$ in $\sigma_i$, then the incidence point of the order-preserving edge $(p'_i,p'_{i+1})$ on $Q_i$ is between the incidence points of $(u'_e,r'_i)$ and $(u'_e,s'_i)$ on $Q_i$. It follows that $(p'_i,p'_{i+1})$ crosses $M_i$, or the matrix representing $u'_e$, or $(u'_e,r'_i)$, or $(u'_e,s'_i)$. 

The case in which $(u'_e,r'_i)$ and $(u'_e,s'_i)$ are both assigned to the side $(t^i_1,y_i)$ and $(b^i_1,w_i)$ are assigned to and $i\geq 2$ is analogous to the previous one. 

This concludes the proof of Lemma~\ref{le:same-side}. 
\end{proof}

\begin{lemma} \label{le:same-ordering}
For any $i,j \in \{1,2,3\}$, the orderings $\sigma_i-\{x_i,y_i,w_i,z_i\}$ of $V'_i$ and $\sigma_j-\{x_j,y_j,w_j,z_j\}$ of $V'_j$ either are the same ordering or are the reverse of each other (according to the bijection between $V'_i$ and $V'_j$). 
\end{lemma}

\begin{proof}
By Claim~\ref{cl:corners}, $(t^i_7,y_i)$ and $(b^i_7,w_i)$ are both assigned to the right or to the left side of $Q_i$, and $(t^{i+1}_1,y_{i+1})$ and $(b^{i+1}_1,w_{i+1})$ are both assigned to the right or to the left side of $Q_{i+1}$. We distinguish four cases, based on these two assignments.

If $(t^i_7,y_i)$ and $(b^i_7,w_i)$ are both assigned to the right side of $Q_i$, and $(t^{i+1}_1,y_{i+1})$ and $(b^{i+1}_1,w_{i+1})$ are both assigned to the left side of $Q_{i+1}$, then the orderings $\sigma_i-\{x_i,y_i,w_i,z_i\}$ of $V'_i$ and $\sigma_{i+1}-\{x_{i+1},y_{i+1},w_{i+1},z_{i+1}\}$ of $V'_{i+1}$ are the same ordering (according to the bijection between $V'_i$ and $V'_{i+1}$). In fact, traversing the right side of $Q_i$ from the incidence point with $(t^i_7,y_i)$ to the incidence point with $(b^i_7,w_i)$, the segments delimiting the rows associated to the vertices in $V'_i$ are encountered in the order $\sigma_i-\{x_i,y_i,w_i,z_i\}$; this is because by Claim~\ref{cl:corners} the vertex $y_i$ precedes the vertex $w_i$ in $\sigma_i$. Analogously, traversing the left side of $Q_{i+1}$ from the incidence point with $(t^{i+1}_1,y_{i+1})$ to the incidence point with $(b^{i+1}_1,w_{i+1})$, the segments delimiting the rows associated to the vertices in $V'_{i+1}$ are encountered in the order $\sigma_{i+1}-\{x_{i+1},y_{i+1},w_{i+1},z_{i+1}\}$. Now if there were a pair of vertices $p'_i$ and $q'_i$ such that $p'_i$ and $q'_i$ are in this order in $\sigma_i$ and such that $q'_{i+1}$ and $p'_{i+1}$ are in this order in $\sigma_{i+1}$, then the order-preserving edges $(p'_i,p'_{i+1})$ and $(q'_i,q'_{i+1})$ would cross each other, contradicting the planarity of $\Gamma$. 

If $(t^i_7,y_i)$ and $(b^i_7,w_i)$ are both assigned to the left side of $Q_i$, and $(t^{i+1}_1,y_{i+1})$ and $(b^{i+1}_1,w_{i+1})$ are both assigned to the right side of $Q_{i+1}$, the proof is analogous. 

If $(t^i_7,y_i)$ and $(b^i_7,w_i)$ are both assigned to the right side of $Q_i$, and $(t^{i+1}_1,y_{i+1})$ and $(b^{i+1}_1,w_{i+1})$ are both assigned to the right side of $Q_{i+1}$, then the orderings $\sigma_i-\{x_i,y_i,w_i,z_i\}$ of $V'_i$ and $\sigma_{i+1}-\{x_{i+1},y_{i+1},w_{i+1},z_{i+1}\}$ of $V'_{i+1}$, respectively, are the reverse of each other (according to the bijection between $V'_i$ and $V'_{i+1}$). In fact, traversing the right side of $Q_i$ from the incidence point with $(t^i_7,y_i)$ to the incidence point with $(b^i_7,w_i)$, the segments delimiting the rows associated to the vertices in $V'_i$ are encountered in the order $\sigma_i-\{x_i,y_i,w_i,z_i\}$. However, traversing the right side of $Q_{i+1}$ from the incidence point with $(t^{i+1}_1,y_{i+1})$ to the incidence point with $(b^{i+1}_1,w_{i+1})$, the segments delimiting the rows associated to the vertices in $V'_{i+1}$ are encountered in the reverse order of $\sigma_{i+1}-\{x_{i+1},y_{i+1},w_{i+1},z_{i+1}\}$; this is because by Claim~\ref{cl:corners} the vertex $y_i$ follows the vertex $w_i$ in $\sigma_i$. Now if there is a pair of vertices $p'_i$ and $q'_i$ such that $p'_i$ and $q'_i$ are in this order in $\sigma_i$ and such that $p'_{i+1}$ and $q'_{i+1}$ are in this order in $\sigma_{i+1}$, then the order-preserving edges $(p'_i,p'_{i+1})$ and $(q'_i,q'_{i+1})$ cross each other, contradicting the planarity of $\Gamma$.  

If $(t^i_7,y_i)$ and $(b^i_7,w_i)$ are both assigned to the left side of $Q_i$, and $(t^{i+1}_1,y_{i+1})$ and $(b^{i+1}_1,w_{i+1})$ are both assigned to the left side of $Q_{i+1}$ the proof is analogous.
\end{proof}

Lemmata~\ref{le:same-side} and~\ref{le:same-ordering} imply that $(V, E)$ is a positive instance of {\sc Partitioned $3$-Page Book Embedding}. In fact, consider the ordering $\cal O$ of $V$ corresponding to $\sigma_1-\{x_1,y_1,w_1,z_1\}$ according to the bijection between $V$ and $V'_1$. By Lemma~\ref{le:same-ordering}, for $i=2,3$, either $\sigma_i -\{x_i,y_i,w_i,z_i\}$ is the same or is the reverse of $\sigma_1-\{x_1,y_1,w_1,z_1\}$ (according to the bijection between $V'_i$ and $V'_1$). Then, for $i=1,2,3$, no two edges $e = (p,q)$ and $f=(r,s)$ in $E_i$ have alternating end-vertices in $\cal O$. Indeed, $e$ and $f$ cannot have alternating end-vertices if $p'_i$ and $q'_i$ are consecutive in $\sigma_i$, or if $r'_i$ and $s'_i$ are consecutive in $\sigma_i$, as in this case $p$ and $q$, or $r$ and $s$ would be consecutive in $\cal O$, respectively. If $p'_i$ and $q'_i$ are not consecutive in $\sigma_i$ and $r'_i$ and $s'_i$ are not consecutive in $\sigma_i$, then by Lemma~\ref{le:same-side} the edges $(u'_e,p'_i)$, $(u'_e,q'_i)$, $(u'_f,r'_i)$, and $(u'_f,s'_i)$ are all assigned to the same side of $Q_i$, hence if $e$ and $f$ had alternating end-vertices then: (i) $(u'_e,p'_i)$ or $(u'_e,q'_i)$ would cross $(u'_f,r'_i)$ or $(u'_f,s'_i)$, or (ii) $(u'_e,p'_i)$ or $(u'_e,q'_i)$ would cross the matrix representing $u'_f$ or $M_i$, or (iii) $(u'_f,r'_i)$ or $(u'_f,s'_i)$ would cross the matrix representing $u'_e$ or $M_i$, or (iv) two among $M_i$ and the matrices representing $u'_e$ and $u'_f$ would cross each other. In each case we would get a contradiction to the planarity of $\Gamma$. This completes the proof of the theorem.
\end{proof}

Let $G=(V,E,\mathcal{C})$ be a flat clustered graph with a given side assignment $s_i$, for each $V_i \in \mathcal{C}$. We say that $G$ is {\em \nt planar with fixed side} if $G$ admits an \nt planar representation $\Gamma$ such that, for every edge $e=(u,v) \in E$ with  $u \in V_i$ and $v \in V_j$, the incidence points of $e$ with the matrices $M_i$ and $M_j$ representing $V_i$ and $V_j$ in $\Gamma$ lie on the straight-line segments corresponding to the $s_i(e)$ side of $M_i$ and to the $s_j(e)$ side of $M_j$, respectively.

\begin{theorem}\label{th:free-ordering-fixed-sides}
{\sc NodeTrix Planarity with Fixed Side} is \NPC even for instances with two clusters.
\end{theorem}


\noindent
\begin{proof}
Lemma~\ref{le:np} will prove that {\sc \nt Planarity with Fixed Side} is in $\mathbb{NP}$. 

\begin{figure}[tb!]
    \centering
      \includegraphics[height=0.32\textheight]{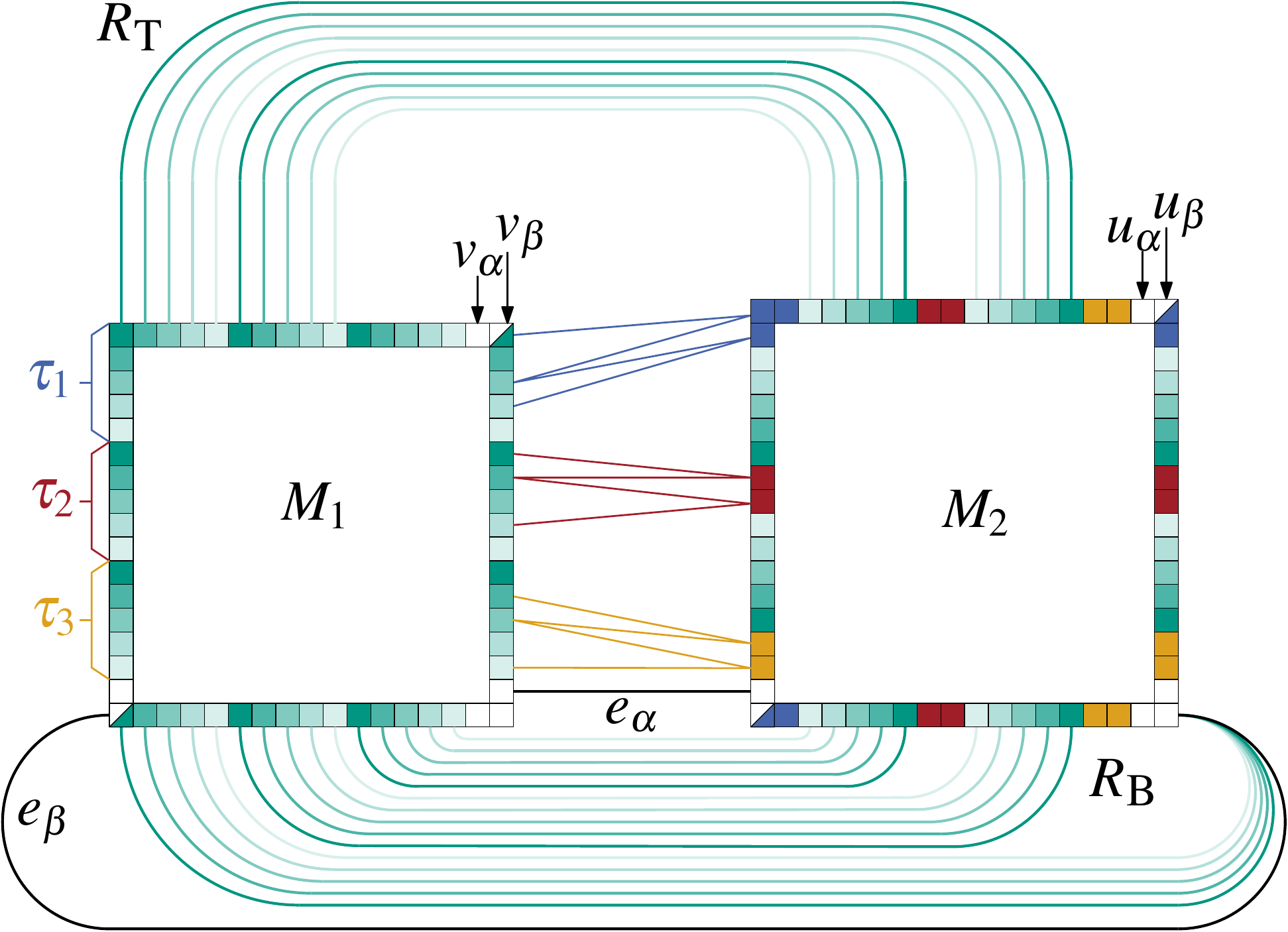}
    \caption{Illustration for the proof of Theorem~\ref{th:free-ordering-fixed-sides}.}
    \label{fi:hardness-free-ordering-fixed-sides}
  \end{figure}

We give a reduction from the \NPC problem {\sc Betweenness}~\cite{o-top-79}, where an instance is a collection of ordered triplets of items and the target is to find a total order of the items in which, for each of the given triplets, the middle item in the triplet appears somewhere between the other two items.

Consider an instance $\cal T$ of {\sc Betweenness}, i.e., a set of $h$ items $\{a_1, a_2, \dots a_h\}$ and a collection of $t$ ordered triplets of $\tau_j = \langle a_{b_j},a_{c_j},a_{d_j} \rangle$, with $j = 1, \dots, t$.
We construct the corresponding instance of {\sc \nt Planarity with Fixed Side} by defining a flat clustered graph $(V,E,\mathcal{C})$, where $\mathcal{C} = \{V_1, V_2\}$, and a side assignment $s_i$, with $i \in \{1,2\}$, as follows: 

\begin{itemize} 

\item cluster $V_1$ contains $2 + h \times t$ vertices: that is, $v_\alpha$, $v_\beta$, and one vertex $v_{[i,j]}$ for each $i=1, \dots, h$ and each $j=1, \dots, t$;

\item cluster $V_2$ contains $2 + h \times (t-1) + 2t$ vertices: that is, $u_\alpha$, $u_\beta$, plus one vertex $u_{[i,j]}$ for each $i=1, \dots, h$ and each $j=1, \dots, t-1$, plus two vertices $u'_j$ and $u''_j$ for each $j=1, \dots, t$.

\end{itemize}

The set $E$ contains an arbitrary set of intra-cluster edges and the following inter-cluster edges.

\begin{itemize} 

\item inter-cluster edge $e_{\alpha} = (v_\alpha, u_\alpha)$, with $s_1(e_\alpha)=\textrm{\sc{r}}$ and $s_2(e_\alpha)=\textrm{\sc{l}}$;

\item inter-cluster edge $e_\beta = (v_\beta, u_\beta)$, with $s_1(e_\beta)=\textrm{\sc{l}}$ and $s_2(e_\beta)=\textrm{\sc{r}}$;

\item for each $i=1, \dots, h$ an inter-cluster edge $e^\textrm{\sc{b}}_{[i,1]} = (u_\beta,v_{[i,1]})$, with $s_1(e^\textrm{\sc{b}}_{[i,1]})=\textrm{\sc{b}}$ and $s_2(e^\textrm{\sc{b}}_{[i,1]}) = \textrm{\sc{r}}$;

\item for each $i=1, \dots, h$ and $j=1, \dots, t-1$ an inter-cluster edge $e^\textrm{\sc{b}}_{[i,j]} = (u_{[i,j]},v_{[i,j+1]})$, with $s_1(e^\textrm{\sc{b}}_{[i,j]}) = s_2(e^\textrm{\sc{b}}_{[i,j]}) = \textrm{\sc{b}}$;

\item for each $i=1, \dots, h$ and $j=1, \dots, t-1$ an inter-cluster edge $e^\textrm{\sc{t}}_{[i,j]} = (v_{[i,j]},u_{[i,j]})$, with $s_1(e^\textrm{\sc{t}}_{[i,j]}) = s_2(e^\textrm{\sc{t}}_{[i,j]}) = \textrm{\sc{t}}$;

\item for each triplet $t_j = \langle a_{b_j}, a_{c_j}, a_{d_j} \rangle$, with $j = 1, \dots, t$, a path of four inter-cluster edges joining the five vertices $v_{[a_{b_j},j]}, u'_j, v_{[a_{c_j},j]}, u''_j$, and $v_{[a_{d_j},j]}$ in this order. Each edge $e$ of such a path has $s_1(e)=\textrm{\sc{r}}$ and $s_2(e)=\textrm{\sc{l}}$.

\end{itemize}

($\Longrightarrow$) Suppose that the items of $\cal T$ admit a total order $a_{\pi_1}, a_{\pi_2}, \dots, a_{\pi_h}$ in which for each of the given triplets, the middle item in the triplet appears somewhere between the other two items. We show how to construct a NodeTrix planar representation of $(V, E, \mathcal{C})$. 

Use for the matrix $M_1$ representing $V_1$ a row-column order $\sigma_1$ such that $\sigma_1(v_{[\pi_1,1]}) < \sigma_1(v_{[\pi_2,1]}) < \dots < \sigma_1(v_{[\pi_h,1]}) < \sigma_1(v_{[\pi_1,2]}) < \dots < \sigma_1(v_{[\pi_h,2]}) < \dots < \sigma_1(v_{[\pi_1,t]}) < \dots < \sigma_1(v_{[\pi_h,t]}) < \sigma_1(e_\alpha) < \sigma_1(e_\beta)$.
Use for the matrix $M_2$ representing $V_2$ a row-column order $\sigma_2$ such that 
$\sigma_2(u'_t) < \sigma_2(u''_t) < \sigma_2(u_{[\pi_h,t-1]}) < \sigma_2(u_{[\pi_{h-1},t-1]}) < \dots < \sigma_2(u_{[\pi_{1},t-1]}) < \sigma_2(u'_{t-1}) < \sigma_2(u''_{t-1}) < \sigma_2(u_{[\pi_h,t-2]})$ $<$ $\sigma_2(u_{[\pi_{h-1},t-2]})$ $< \dots < \sigma_2(u_{[\pi_{1},t-2]}) < \sigma_2(u'_{t-2}) < \sigma_2(u''_{t-2}) < \dots < 
\sigma_2(u_{[\pi_h,1]}) < \sigma_2(u_{[\pi_{h-1},1]}) < \dots < \sigma_2(u_{[\pi_{1},1]}) < \sigma_2(u'_1) < \sigma_2(u''_1) < \sigma_2(e_\alpha) < \sigma_2(e_\beta)$. It can be easily seen that the inter-cluster edges can be drawn attached to the sides imposed by $s_1$ and $s_2$ without crossings, as in Fig.~\ref{fi:hardness-free-ordering-fixed-sides}.

($\Longleftarrow$) Suppose that $(V, E, \mathcal{C})$ admits a NodeTrix planar representation where, for $i \in \{1,2\}$, each inter-cluster edge attaches according to the edge assignment $s_i$ to the matrix $M_i$ representing the cluster $V_i$. We show that $\cal T$ admits a total order in which for each triplet, the middle item in the triplet appears somewhere between the other two items. 

First observe that, whatever the row-column orders $\sigma_1$ and $\sigma_2$ chosen for matrices $M_1$ and $M_2$ are, respectively, the matrices $M_1$ and $M_2$ form, together with the edges $e_\alpha$ and $e_\beta$, a cycle that separates the top sides of the two matrices from their bottom sides. It follows that all inter-cluster edges that attach to the top side (to the bottom side) of $M_1$ or $M_2$ are drawn inside the same region delimited by the boundaries of $M_1$ and $M_2$, by $e_\alpha$, and by $e_\beta$; we denote by $R_\textrm{\sc t}$ (by $R_\textrm{\sc b}$) the one of these regions comprising the top side (resp.\ the bottom side) of $M_1$ and $M_2$. Refer again to Fig.~\ref{fi:hardness-free-ordering-fixed-sides} for an illustration.

Consider the inter-cluster edges $e^\textrm{\sc{b}}_{[i,1]}$, for $i =1, \dots, h$. Since they all attach to $u_\beta$ and since $s_1(e^\textrm{\sc{b}}_{[i,1]})=\textrm{\sc{b}}$, they are all drawn inside $R_\textrm{\sc b}$. Denote by $\pi = \pi_1, \dots, \pi_h$ the permutation of the indices $1, \dots, h$ such that $\sigma_1(v_{[\pi_1,1]}) < \sigma_1(v_{[\pi_2,1]}) < \dots < \sigma_2(v_{[\pi_h,1]})$ (recall that these are the end-vertices of the edges $e^\textrm{\sc{b}}_{[i,1]}$). 

We claim that $\sigma_1(v_{[\pi_1,1]}) < \dots < \sigma_1(v_{[\pi_h,1]}) < \sigma_1(v_{[\pi_1,2]}) < \dots < \sigma_1(v_{[\pi_h,2]}) < \dots < \sigma_1(v_{[\pi_1,t]}) < \dots \sigma_1(v_{[\pi_h,t]})$ holds true. First, we prove that, for each $j\in \{1, \dots, t\}$, the vertices $v_{[\pi_1,j]}, v_{[\pi_2,j]}, \dots, v_{[\pi_h,j]}$ appear in this order in the row-column order of $M_1$; indeed, by the definition of $\pi$, this is the case for $j=1$. Observe that the inter-cluster edges $e^\textrm{\sc{t}}_{[i,1]}$, for $i=1, \dots, h$, are drawn inside $R_\textrm{\sc t}$ and force $\sigma_2$ to be such that $\sigma_2(u_{[\pi_h,1]}) < \sigma_2(u_{[\pi_{h-1},1]}) < \dots < \sigma_2(u_{[\pi_1,1]})$ (recall that the end-vertices of the edge $e^\textrm{\sc{t}}_{[i,1]}$ are $u_{[i,1]}$ and $v_{[i,1]}$). Analogously, the inter-cluster edges $e^\textrm{\sc{b}}_{[i,1]}$, for $i = 1, \dots, h$, are drawn inside $R_\textrm{\sc b}$ and force $\sigma_1$ to be such that $\sigma_1(v_{[\pi_1,2]}) < \sigma_1(v_{[\pi_2,2]}) < \dots < \sigma_1(v_{[\pi_h,2]})$ (recall that the end-vertices of the edge $e^\textrm{\sc{b}}_{[i,1]}$ are $u_{[i,1]}$ and $v_{[i,2]}$). For $j=2, \dots, t-1$, the same argument can be repeated alternately for all the inter-cluster edges $e^\textrm{\sc{t}}_{[i,j]}$ and then for all the inter-cluster edges $e^\textrm{\sc{b}}_{[i,j]}$; then for any $j \in \{1, \dots t\}$, it holds true that $\sigma(v_{[\pi_1,j]}) < \sigma(v_{[\pi_2,j]}) < \dots < \sigma_1(v_{[\pi_h,j]})$. Also, it is easy to see that, in order for the drawing to be crossing-free, $\sigma_1(v_{[i',j']}) < \sigma_1(v_{[i'',j'']})$ whenever $j' < j''$. This concludes the proof of the claim.

Now, for each $j = 1, \dots, t$, consider the inter-cluster edges of the path $v_{[a_{b_j},j]}, u'_j,$ $v_{[a_{c_j},j]}, u''_j$, and $v_{[a_{d_j},j]}$. In any \nt planar representation of $(V, E, \mathcal{C})$ such a path forces $v_{[a_{c_j},j]}$ to be in the middle of $v_{[a_{b_j},j]}$ and $v_{[a_{d_j},j]}$, that is, it forces $a_{c_j}$ to be in the middle of $a_{b_j}$ and $a_{d_j}$ in $\pi$. It follows that $\pi$ is an ordering of the items $a_{\pi_1}, a_{\pi_2}, \dots, a_{\pi_h}$ in which, for each triplet, the middle item in the triplet appears between the other two items.
\end{proof}

Let $G=(V,E,\mathcal{C})$ be a flat clustered graph with a given row-column order $\sigma_i$, for each $V_i \in \mathcal{C}$. We say that $G$ is {\em \nt planar with fixed order} if it admits an \nt planar representation $\Gamma$ where, for each cluster $V_i \in \mathcal{C}$, each vertex $v \in V_i$ is associated with the $\sigma_i(v)$-th row and column of the matrix $M_i$ representing $V_i$ in $\Gamma$.

\begin{theorem}\label{th:fixed-ordering-free-sides}
{\sc NodeTrix Planarity with Fixed Order} is \NPC even if at most one cluster contains more than one vertex.
\end{theorem}


\noindent
\begin{proof}
The membership in $\mathbb{NP}$ of {\sc \nt Planarity with Fixed Order} will be proved in Lemma~\ref{le:np}. 

For the $\mathbb{NP}$-hardness, we give a reduction from the \NPC problem that asks to determine whether a proper $4$-coloring exists for a {\em circle graph}~\cite{u-okccg-88}, which is an intersection graph of chords of a circle. Let $G=(N,A)$ be a circle graph. First, by means of the algorithm in~\cite{spinrad-rcg-94}, we construct in polynomial time an intersection representation $\langle \mathcal{P}, \mathcal{O} \rangle$ of $G$, where $\mathcal{P}$ is a linear sequence of distinct points on a circle and $\mathcal{O}$ is a set of chords between pairs of points in $\mathcal{P}$ such that: (i) each chord $c \in \mathcal{O}$ corresponds to a vertex $n \in N$ and (ii) two chords $c',c''\in \mathcal{O}$ intersect if and only if $(n',n'') \in A$, where $n'$ and $n''$ are the vertices in $N$ corresponding to $c'$ and $c''$, respectively. Then $G$ admits a proper $4$-coloring if and only if the chords in $\mathcal{O}$ can be $4$-colored so that no two chords of the same color intersect. Starting from $\langle \mathcal{P}, \mathcal{O} \rangle$ we construct an instance $(V,E,\mathcal{C})$ of {\sc NodeTrix Planarity with Fixed Order} as follows (refer to Fig.~\ref{fi:hardness-prescribedorder}). Set $\cal C$ contains:


  \begin{figure}[tb!]
    \centering
    \subfloat[]{
      \includegraphics[height=0.2\textheight,page=2]{img/hardness-fixedorder}
      \label{fi:hardness-prescribedorder-a}
    }\hfil
    \subfloat[]{
      \includegraphics[height=0.2\textheight,page=1]{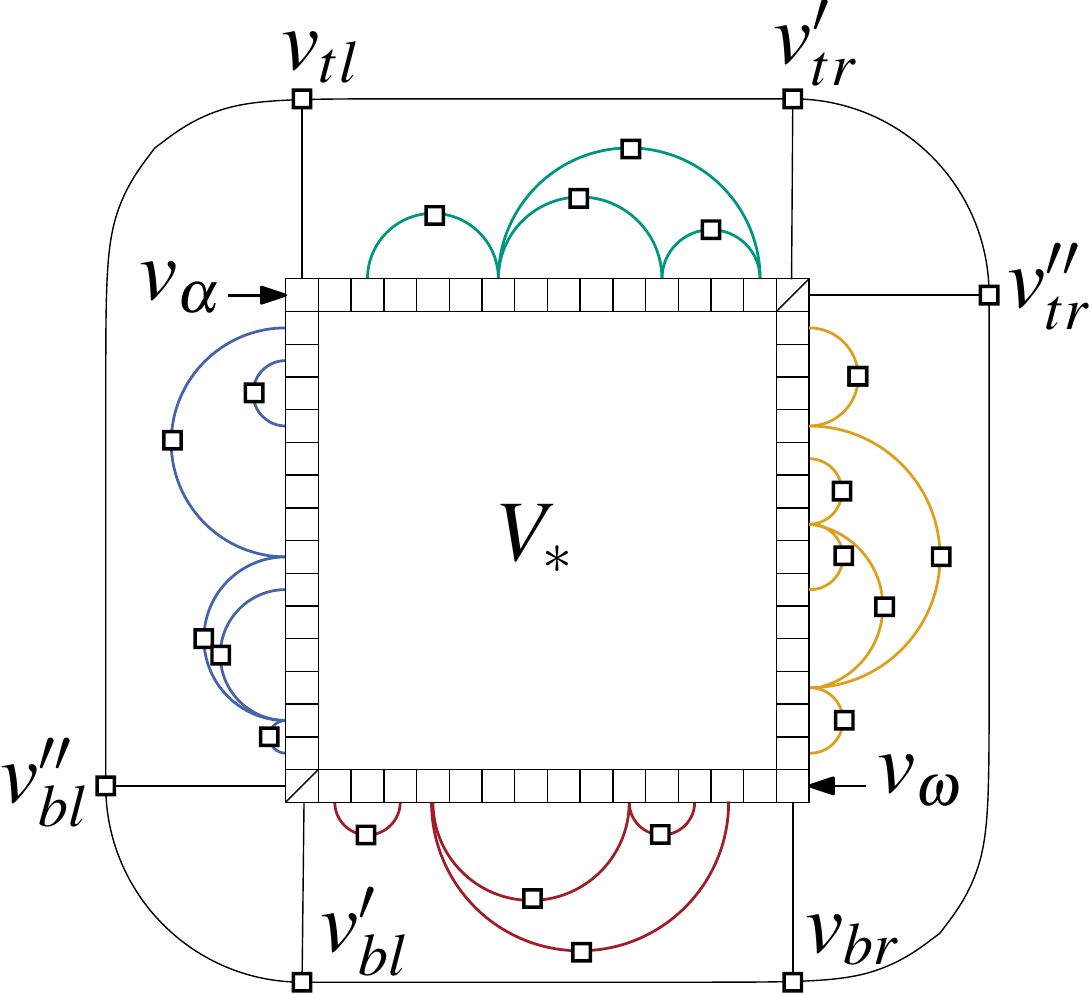}
      \label{fi:hardness-prescribedorder-b}
    }
    \caption{
      (a) An intersection representation $\langle \mathcal{P}, \mathcal{O} \rangle$  of a circle graph $G=(N,A)$. (b) Instance $(V,E,\mathcal{C})$ of {\sc NodeTrix Planarity} corresponding to $\langle \mathcal{P}, \mathcal{O} \rangle$.
    }
    \label{fi:hardness-prescribedorder}
  \end{figure}

\begin{itemize}
\item a cluster $V_*$ containing one vertex $v_i$ for each point $p_i \in \mathcal{P}$, plus two additional vertices $v_\alpha$ and $v_\omega$;
\item six clusters $\{v_{tl}\}$, $\{v'_{tr}\}$, $\{v''_{tr}\}$, $\{v_{br}\}$, $\{v'_{bl}\}$, and $\{v''_{bl}\}$, respectively; and
\item a cluster $\{v_c\}$, for each chord $c \in \mathcal{O}$.
\end{itemize}

The set $E$ contains an arbitrary set of intra-cluster edges and the following inter-cluster edges.

\begin{itemize} 
\item {\em bounding edges} $(v_{tl},v'_{tr})$, $(v'_{tr},v''_{tr})$, $(v''_{tr},v_{br})$, $(v_{br},v'_{bl})$, $(v'_{bl},v''_{bl})$, and $(v''_{bl}, v_{tl})$; 
\item {\em corner edges}  $(v_{tl},v_\alpha), (v'_{tr},v_\omega), (v''_{tr},v_\alpha), (v_{br},v_\omega)$, $(v'_{bl},v_\alpha)$, and $(v''_{bl},v_\omega)$; and 
\item {\em chord edges}: for each chord $c = (p_i,p_j) \in \mathcal{O}$, edges $(v_i,v_c)$ and $(v_c,v_j)$. 
\end{itemize}

Finally, we fix the row-column order $\sigma_{*}$ of the only non-unitary cluster $V_*$ to be $v_\alpha \circ \mathcal{P} \circ v_\omega$ (where with a slight abuse of notation we denote by $\mathcal{P}$ not only the order of the points on the circle in the given intersection representation of $G$, but also the corresponding order of the vertices in $V_* -\{v_\alpha,v_\omega\}$). We now prove the equivalence between the problem of properly $4$-coloring $\langle \mathcal{P}, \mathcal{O} \rangle$ and the constructed instance of {\sc NodeTrix Planarity with Fixed Order}.

($\Longrightarrow$) Suppose that the chords of $\langle \mathcal{P}, \mathcal{O} \rangle$ can be assigned colors $1,2,3,4$ so that no two chords with the same color intersect. We show how to construct a NodeTrix planar representation with fixed order of $(V, E, \mathcal{C})$. Represent clusters $V_*$, $\{v_{tl}\}$, $\{v'_{tr}\}$, $\{v''_{tr}\}$, $\{v_{br}\}$, $\{v'_{bl}\}$, and $\{v''_{bl}\}$ by matrices $M_*$, $M_{tl}$, $M'_{tr}$, $M''_{tr}$, $M_{br}$, $M'_{bl}$, and $M''_{bl}$, respectively, where the row-column order of $M_*$ is $\sigma_*$. Draw the bounding edges so that $M_*$ is inside the cycle $D$ they compose together with $M_{tl}$, $M'_{tr}$, $M''_{tr}$, $M_{br}$, $M'_{bl}$, and $M''_{bl}$. Draw the corner edges also inside $D$. The corner edges, together with the boundary of $M_*$, subdivide the region of the plane inside $D$ into five regions, namely one region internal to the boundary of $M_*$ and four regions incident to the top, right, bottom and left side of $M_*$. We refer to these regions as to the {\em top}, {\em right}, {\em bottom}, and {\em left region}, respectively. Depending on whether a chord $c=(p_i,p_j)$ has color $1$, $2$, $3$, or $4$, we draw the chord edges $(v_i, v_c)$ and $(v_c,v_j)$, as well as the matrix $M_c$ representing cluster $\{v_c\}$, inside the top, right, bottom, and left region, respectively. We claim that the obtained \nt representation of $(V, E, \mathcal{C})$ is planar. Suppose, for a contradiction, there is a crossing between two paths $(v_i, v_{c'}, v_j)$ and $(v_k, v_{c''}, v_h)$ corresponding to two chords $c'=(p_i,p_j)$ and $c''=(p_k,p_h)$. Then these two paths are in the interior of the same (top, right, bottom, or left) region, hence they attach to the same side of $M_*$; it follows that the two chords $c'$ and $c''$ have the same color. By the definition of $\sigma_*$, since the end-vertices of the two paths alternate along the side of $M_*$, the end-points of $c'$ and $c''$ alternate in $\cal P$. Hence, $c'$ and $c''$ cross, thus contradicting the fact that $c'$ and $c''$ have the same color.

($\Longleftarrow$) Suppose that $(V, E, \mathcal{C})$ admits a NodeTrix-planar representation $\Gamma$ with a row-column order $\sigma_*$ for the unique non-unitary cluster $V_*$. Denote by $M_*$ the matrix representing $V_*$ in $\Gamma$. We show that the chords of $\langle \mathcal{P}, \mathcal{O} \rangle$ are $4$-colorable so that no two chords of the same color intersect. 

Similarly to the proof of Theorem~\ref{th:free-ordering-free-sides}, it can be proved that the corner edges subdivide the side of $D$ that contains $M_*$ into five regions, defined as in the previous direction, so that all the vertices on the top, right, bottom, and left side of $M_*$ are incident to the top, right, bottom, and left region, respectively; while in the proof of Theorem~\ref{th:free-ordering-free-sides} this was ensured by Claim~~\ref{cl:corners}, it is here a trivial consequence of the fact that $v_{\alpha}$ and $v_{\omega}$ are the first and the last vertex in $\sigma_*$. By the planarity of $\Gamma$, both the incidence points of a path $(v_i, v_c, v_j)$ with the boundary of $M_*$ are on the same side of $M_*$. Then color all the chords $c=(p_i,p_j)$ in $\cal O$ such that path $(v_i, v_c, v_j)$ is in the top, right, bottom, or left region with color $1$, $2$, $3$, or $4$, respectively.

We claim that the obtained $4$-coloring of the chords of $\langle \mathcal{P}, \mathcal{O} \rangle$ is proper. Suppose, for a contradiction, there is a crossing in $\langle \mathcal{P}, \mathcal{O} \rangle$ between two chords $c'=(p_i,p_j)$ and $c''=(p_k,p_h)$ both with color $1$ -- the discussion for the other colors is analogous. Then $c'$ and $c''$ have alternating end-points in $\langle \mathcal{P}, \mathcal{O} \rangle$. Since the order of the points in $\cal P$ coincides with the order of the corresponding vertices in $\sigma_*$, it follows that paths $(v_i, v_{c'}, v_j)$ and $(v_k, v_{c''}, v_h)$ have alternating end-points on the top side of $M_*$, hence they cross, a contradiction to the planarity of $\Gamma$. This concludes the proof.
\end{proof}

Let $G=(V,E,\mathcal{C})$ be a flat clustered graph with a given row-column order $\sigma_i$ and side assignment $s_i$, for each $V_i \in \mathcal{C}$. Then $G$ is {\em \nt planar with fixed order and fixed side} if it is simultaneously planar with fixed order and with fixed side. 

\begin{theorem}\label{th:fixed-ordering-fixed-sides}
{\sc NodeTrix Planarity with Fixed Order and Fixed Side} can be solved in linear time.
\end{theorem}

\begin{proof}
Consider the graph $G'$ obtained from an instance $G=(V,E,\mathcal{C})$ of {\sc NodeTrix Planarity with Fixed Order and Fixed Side} by collapsing each cluster $V_i \in \mathcal{C}$ into a vertex~$v_i$. Intuitively, instance $G$ is \nt planar with fixed order and fixed side if and only if $G'$ is planar with the additional constraint that the clockwise order of the edges incident to each vertex $v_i$ is ``compatible'' with the row-column order $\sigma_i$ and the side assignment $s_i$ for the cluster $V_i$.

More formally, denote by ${\cal E}_i$ the set of the inter-cluster edges incident to $V_i$ and denote by $v_i(e)$ the vertex of $V_i$ incident to an edge $e \in {\cal E}_i$. The edges in ${\cal E}_i$ can be decomposed into a circular sequence of sets ${\cal S} = {\cal E}_{{\textrm{\sc{t}}},1}, {\cal E}_{{\textrm{\sc{t}}},2}, \dots, {\cal E}_{{\textrm{\sc{t}}},|V_i|},$  
${\cal E}_{{\textrm{\sc{r}}},1}, {\cal E}_{{\textrm{\sc{r}}},2}, \dots, {\cal E}_{{\textrm{\sc{r}}},|V_i|},$  
${\cal E}_{{\textrm{\sc{b}}},|V_i|}, {\cal E}_{{\textrm{\sc{b}}},|V_i-1|}, \dots, {\cal E}_{{\textrm{\sc{b}}},1},$  
${\cal E}_{{\textrm{\sc{l}}},|V_i|}, {\cal E}_{{\textrm{\sc{l}}},|V_i-1|}, \dots, {\cal E}_{{\textrm{\sc{l}}},1}$, where each ${\cal E}_{{\textrm{\sc{x}}},j}$, with $\textrm{\sc{x}} \in \{\textrm{\sc{t}}, \textrm{\sc{b}},$ $\textrm{\sc{l}}, \textrm{\sc{r}}\}$ and $j \in \{1, \dots, |V_i|\}$, contains the edges $e\in {\cal E}_i$ such that $s_i(e)=\textrm{\sc{x}}$ and $\sigma_i(v_i(e))=j$. Let $\Gamma'$ be a planar embedding of $G'$ and let $\lambda_i$ denote the clockwise order of the edges incident to vertex $v_i$ of $G'$ in $\Gamma'$. 
The embedding $\Gamma'$ of $G'$ is \emph{compatible} with functions $\sigma_i$ and $s_i$ if: (i) all the edges belonging to the same set ${\cal E}_{{\textrm{\sc{x}}},j}$ appear consecutively in $\lambda_i$, and (ii) for any three edges $e' \in {\cal E}_{{\textrm{\sc{x}}'},j'}$,  $e'' \in {\cal E}_{{\textrm{\sc{x}}''},j''}$, and $e''' \in {\cal E}_{{\textrm{\sc{x}}'''},j'''}$, where ${\cal E}_{{\textrm{\sc{x}}'},j'}$, ${\cal E}_{{\textrm{\sc{x}}''},j''}$, and ${\cal E}_{{\textrm{\sc{x}}'''},j'''}$ are all distinct, appear in this clockwise order in $\lambda_i$ if and only if ${\cal E}_{{\textrm{\sc{x}}'},j'}$, ${\cal E}_{{\textrm{\sc{x}}''},j''}$, and ${\cal E}_{{\textrm{\sc{x}}'''},j'''}$ appear in this circular order in $\cal S$.

It can be easily seen that an instance of {\sc NodeTrix Planarity with Fixed Order and Fixed Side} has a solution if and only if $G'$ admits an embedding $\Gamma'$ that is compatible with $\sigma_i$ and $s_i$, for all vertices $v_i$ of $G'$.  
We obtain an instance of constrained planarity for $G'$ that can be tested in linear time with known techniques~\cite{gkm-ptoeiec-08}.
\end{proof}


We conclude the section with the following lemma.

\begin{lemma} \label{le:np}
{\sc NodeTrix Planarity}, {\sc NodeTrix Planarity with Fixed Side}, and {\sc NodeTrix Planarity with Fixed Order} are in \NP.
\end{lemma}

\remove{ 
\begin{sketch}
The number of distinct row-column orders and side assignments for an instance $(V,E,\mathcal{C})$ is a function of $|V|+|E|$. The statement follows from Theorem~\ref{th:fixed-ordering-fixed-sides}.
\end{sketch}
} 

\noindent
\begin{proof}
We prove the statement for {\sc \nt Planarity}; the other proofs are analogous. Consider an instance $(V,E,\mathcal{C})$ of {\sc \nt Planarity}. For each $V_i \in \mathcal{C}$, guess a row-column order $\sigma_i$ and a side assignment $s_i$; then use the algorithm described in the proof of Theorem~\ref{th:fixed-ordering-fixed-sides} to test in linear time whether $(V,E,\mathcal{C})$ is \nt planar with fixed order and fixed side. Since the number of distinct row-column orders and side assignments is a function of $|V|+|E|$, we get the \NP~membership.
\end{proof}

\section{Monotone NodeTrix Representations}\label{se:monotone}

Let $G=(V,E,\mathcal{C})$ be a flat clustered graph and $\gamma$ be a square assignment for $G$ that maps each cluster in $\cal C$ to an axis-aligned square in the plane. A curve is {\em $x$-monotone} (resp.\ {\em $y$-monotone}) if no two of its points have the same projection on the $x$-axis (resp.\ on the $y$-axis) and is {\em $xy$-monotone} if it is either a horizontal or a vertical segment or it is both $x$- and $y$-monotone.  A {\em monotone} \nt representation $\Gamma$ of $\langle G, \gamma \rangle$ is an \nt representation such that: 
\begin{enumerate}
\item all the inter-cluster edges are represented by $xy$-monotone curves; 
\item for each cluster $V_i \in {\cal C}$, the boundary of the matrix $M_i$ representing $V_i$ is $Q_i=\gamma(V_i)$; 
\item for each pair of adjacent clusters $V_i$ and $V_j$, with $i\neq j$, the convex hull of $Q_i$ and $Q_j$ does not intersect any other square $Q_k$, with $k\neq i, j$ -- we call this convex hull the \emph{pipe} of $Q_i$ and $Q_j$; and 
\item all the inter-cluster edges between vertices in $V_i$ and vertices in $V_j$ lie inside the pipe of $Q_i$ and $Q_j$.
\end{enumerate}
In a monotone \nt representation $\Gamma$ of $G$ let $\chi_i(\Gamma)$ denote the number of edge crossings between pairs of inter-cluster edges incident to $V_i$. Let $\chi(\Gamma) = \sum_{i} \chi_i(\Gamma)$, where the sum is over all the clusters $V_i \in {\cal C}$; we say that $\Gamma$ is {\em locally planar} if $\chi(\Gamma)=0$ and no inter-cluster edge intersects any matrix except at its incidence points. The notions of {\em fixed order} and {\em fixed side} easily extend to monotone \nt representations. 

We study the complexity of testing if a flat clustered graph with fixed square assignment admits a monotone locally-planar \nt representation, a problem which we call {\sc Monotone \nt Local Planarity} ({\sc MNTLP}). The next two theorems show the \NP-hardness of {\sc MNTLP} and of its variant with fixed side.

\begin{theorem}\label{th:monotone-free-ordering-free-sides}
{\sc MNTLP} is~\NPC.
\end{theorem}

\noindent
\begin{proof}
The proof that the problem is in \NP is similar to the proof of Theorem~\ref{th:fixed-ordering-fixed-sides}: one can guess a row-column order and a side assignment for each cluster; then the monotone \nt local planarity of the given clustered graph with the given square assignment and the guessed row-column order and side assignment can be tested in polynomial time by Theorem~\ref{th:monotone-fixed-order-and-side-polynomial}, to be presented later.

For the \NP-hardness we give a reduction from the \NPC problem {\sc Betweenness}~\cite{o-top-79}, defined in the proof of Theorem~\ref{th:free-ordering-fixed-sides}. Consider an instance $\cal T$ of {\sc Betweenness} consisting of a set of $h$ items $\{a_1, \dots a_h\}$ and $m$ ordered triplets $\tau_i = \langle a_{b_i},a_{c_i},a_{d_i} \rangle$, with $i = 1, \dots, m$. We construct an instance $\langle G=(V,E,\mathcal{C}),\gamma \rangle$ of {\sc MNTLP} as follows (refer to Fig.~\ref{fig:monotone-hardness}). Let $\cal C$ consist of:

  \begin{figure}[tb!]
    \centering
    \includegraphics[width=0.8\textwidth,page=1]{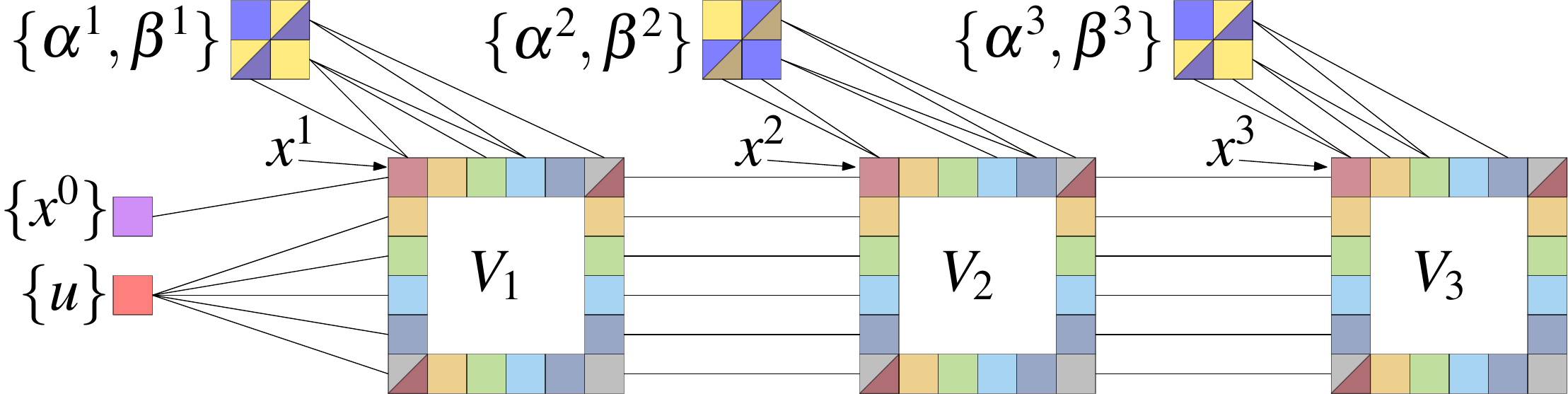}
    \caption{Illustration for the proof of Theorem~\ref{th:monotone-free-ordering-free-sides}, with $m=3$.}
  \label{fig:monotone-hardness}
  \end{figure}
  
\begin{itemize}
\item for $i=1,\dots,m$, a cluster $V_i$ containing $h+1$ vertices $v^i_1,\dots,v^i_h,x^i$;
\item two clusters $\{x^0\}$ and $\{u\}$; and 
\item for $i=1,\dots,m$, a cluster $\{\alpha^i,\beta^i\}$.
\end{itemize}

The set $E$ contains an arbitrary set of intra-cluster edges and the following inter-cluster edges.

\begin{itemize}
\item {\em order-preserving edges} connecting $v^i_j$ with $v^{i+1}_j$, for $1\leq i\leq m-1$ and $1\leq j\leq h$;
\item {\em side-filling edges} connecting $u$ with $v^1_j$, for $1\leq j\leq h$;
\item {\em protecting edges} connecting $x^i$ with $x^{i+1}$, for $0\leq i\leq m-1$; 
\item {\em corner edges} connecting $x^i$ with $\alpha^i$ and $\beta^i$, for $1\leq i\leq m$; and 
\item {\em betweenness edges} connecting, for each triplet $\tau_i = \langle a_{b_i},a_{c_i},a_{d_i} \rangle$ with $i\in \{1, \dots, m\}$, the vertex $v^i_{b_i}$ with $\alpha^i$, the vertex $v^i_{d_i}$ with $\beta^i$, and the vertex $v^i_{c_i}$ with both $\alpha^i$ and $\beta^i$. 
\end{itemize}

Finally, square assignment $\gamma$ is defined as follows:

\begin{itemize} 
\item for $i=1,\dots,m$, the cluster $V_i$ is assigned to a square $Q_i$ so that, for every $1\leq i < j\leq m$, we have $\min_y\{Q_i\}=\min_y\{Q_j\}$, $\max_y\{Q_i\}=\max_y\{Q_j\}$, and $\max_x\{Q_i\}<\min_x\{Q_j\}$; 
\item the cluster $\{x^0\}$ (the cluster $\{u\}$) is assigned to a square $Q_0$ (resp.\ $Q_u$), so that $\min_x\{Q_0\}=\min_x\{Q_u\}<\max_x\{Q_0\}=\max_x\{Q_u\}<\min_x\{Q_1\}$, and $\min_y\{Q_1\}<\min_y\{Q_u\}<\max_y\{Q_u\}<\min_y\{Q_0\}<\max_y\{Q_0\}<\max_y\{Q_1\}$; and 
\item for $i=1,\dots,m$, the cluster $\{\alpha^i,\beta^i\}$ is assigned to a square $Q^{\alpha}_i$ so that $\max_x\{Q_{i-1}\}<\min_x\{Q^{\alpha}_{i}\}<\max_x\{Q^{\alpha}_{i}\}<\min_x\{Q_{i}\}$, and $\max_y\{Q_i\}<\min_y\{Q^{\alpha}_i\}$.
\end{itemize} 

Notice that the above definition of $\gamma$ is such that, for each pair of adjacent clusters $V_i$ and $V_j$, with $i\neq j$, the convex hull of $Q_i$ and $Q_j$ does not intersect any other square $Q_k$, with $k\neq i, j$. 

We now prove the equivalence between the given instance of {\sc Betweenness} and the constructed instance of {\sc MNTLP}. 

($\Longrightarrow$) If $\cal T$ admits an order $\sigma=(a_{\pi_1}, \dots, a_{\pi_h})$ in which $a_{c_i}$ appears between $a_{b_i}$ and $a_{d_i}$ in $\sigma$, for each triplet $\tau_i = \langle a_{b_i},a_{c_i},a_{d_i} \rangle$, we construct a monotone locally-planar \nt representation $\Gamma$ with fixed square assignment of $G$ as follows. For $i=1,\dots,m$, represent $V_i$ as a matrix $M_i$ with boundary $Q_i$ and with row-column order $x^i,v^i_{\pi_1},\dots,v^i_{\pi_h}$; represent $\{\alpha^i,\beta^i\}$ as a matrix $M^{\alpha}_i$ with boundary $Q^{\alpha}_i$ and with row-column order $\alpha^i,\beta^i$ or $\beta^i,\alpha^i$ depending on whether $v^i_{b_i}$ follows or precedes $v^i_{d_i}$ in $\sigma$, respectively; the representation of the unitary clusters $\{x^0\}$ and $\{u\}$ in $\Gamma$ is trivially defined. Assign every inter-cluster edge incident to $x^0$ (or $u$) to the right side of $Q_0$ (or $Q_u$) and to the left side of $Q_1$, every order-preserving or protecting edge between a vertex in $V_i$ and a vertex in $V_{i+1}$ to the right side of $Q_i$ and to the left side of $Q_{i+1}$, every corner edge incident to $x^i$ to the top side of $Q_i$ and to the bottom side of $Q^{\alpha}_i$, and every betweenness edge incident to a vertex in $V_i$ to the top side of $Q_i$ and to the right side of $Q^{\alpha}_i$. Finally, draw all the inter-cluster edges as straight-line segments in $\Gamma$. Observe that such segments are $xy$-monotone curves inside the corresponding pipe.

Representation $\Gamma$ has no crossing between any inter-cluster edge and any matrix, as a consequence of the square and side assignments and independently of the row-column order of the matrices. Further, the order-preserving, side-filling, and protecting edges do not cross the corner and betweenness edges since they are separated by the horizontal line $y=\max_y\{Q_1\}$, and do not cross each other since the row-column orders of any two matrices $M_i$ and $M_{i+1}$ both correspond to $\sigma$ with $x^i$ and $x^{i+1}$ as the first element, respectively. The corner edges incident to $x^i$ do not cross the betweenness edges incident to vertices in $V_i$, because of the side assignment of these edges to $Q^{\alpha}_i$ and since $x^i$ precedes $v^i_{\pi_1},\dots,v^i_{\pi_h}$ in the row-column order of $M_i$. Finally, the betweenness edges do not cross each other because of the row-column order defined for $M^{\alpha}_i$ and since the top side of the column representing $v^i_{c_i}$ in $M_i$ is between the top sides of the columns representing $v^i_{b_i}$ and $v^i_{d_i}$ in $M_i$.

($\Longleftarrow$) Suppose that $G$ admits a monotone locally-planar \nt representation $\Gamma$ with fixed square assignment $\gamma$. Let $M_1,\dots,M_m$ be the matrices representing the clusters $V_1,\dots,V_m$ in $\Gamma$, respectively. First, the monotonicity of $\Gamma$ and the placement of squares $Q_0$, $Q_1$, and $Q_u$ imply that the side-filling edges, as well as the edge $(x^0,x^1)$, lie to the right of $Q_0$ and $Q_u$, to the left of $Q_1$, above the line $y=\min_y\{Q_1\}$, and below the line $y=\max_y\{Q_1\}$. Since $\max_y\{Q_u\}<\min_y\{Q_0\}$, we have that all the side-filling edges lie below edge $(x^0,x^1)$, hence the planarity of $\Gamma$ implies that $x^1$ is the first vertex in the row-column order of $M_1$; let $x^1,v^1_{\pi_1},\dots,v^1_{\pi_h}$ be such an order, for some permutation $\pi_1,\dots,\pi_h$ of $\{1,\dots,h\}$. We claim that the total ordering $\sigma=(a_{\pi_1}, \dots, a_{\pi_h})$ is a solution to instance $\cal T$ of {\sc betweenness}.

We first prove that the order $x^1,v^1_{\pi_1},\dots,v^1_{\pi_h}$ is ``preserved'' in $M_2,\dots,M_m$. The monotonicity of $\Gamma$ and the placement of squares $Q_i$ imply that, for $1\leq i\leq m-1$, the order-preserving and protecting edges between vertices in $V_i$ and vertices in $V_{i+1}$ are to the right of $Q_i$, to the left of $Q_{i+1}$, above the line $y=\min_y\{Q_i\}$, and below the line $y=\max_y\{Q_i\}$. Then the planarity of $\Gamma$ implies that, if the row-column order of $M_i$ is $x^i,v^i_{\pi_1},\dots,v^i_{\pi_h}$, the row-column order of $M_{i+1}$ is $x^{i+1},v^{i+1}_{\pi_1},\dots,v^{i+1}_{\pi_h}$. 

Now consider any triplet $\tau_i = \langle a_{b_i},a_{c_i},a_{d_i} \rangle$ in $\cal T$. The monotonicity of $\Gamma$ and the placement of the squares $Q_i$ and $Q^{\alpha}_i$ imply that every corner or betweenness edge is assigned to the top or left side of $Q_i$ and to the bottom or right side of $Q^{\alpha}_i$. Further, since $x^i$ is the first element in the row-column order of $M_i$ and since $\max_y\{Q_i\}<\min_y\{Q^{\alpha}_i\}$, no betweenness edge is assigned to the left side of $Q_i$, as otherwise it would cross edge $(x^{i-1},x^{i})$. Hence, all the betweenness edges are assigned to the top side of $Q_i$. Since $x^i$ is the first vertex in the row-column order of $M_i$ and by the planarity of $\Gamma$, we have that, when traversing $Q^{\alpha}_i$ in clockwise direction starting from its top-right corner, the incidence points between $Q^{\alpha}_i$ and the betweenness edges are encountered before the incidence points between $Q^{\alpha}_i$ and the corner edges; in particular, no betweenness edge incident to the first vertex -- say $\alpha^i$ as other case is analogous -- in the row-column order of $M^{\alpha}_i$ is assigned to the bottom side of $M^{\alpha}_i$, as otherwise this edge would cross the corner edge incident to $\beta^i$. The planarity of $\Gamma$ also implies that, when traversing the top side of $Q_i$ from left to right, the end-points of the betweenness edges incident to $\beta^i$ are encountered all before the end-points of the betweenness edges incident to $\alpha^i$. Since $v^i_{c_i}$ is the only vertex among $v^i_{b_i}$, $v^i_{c_i}$, and $v^i_{d_i}$ that is neighbor of both $\alpha^i$ and $\beta^i$, then its associated column is between the columns associated to $v^i_{b_i}$ and $v^i_{d_i}$, hence $a_{c_i}$ is between $a_{b_i}$ and $a_{d_i}$ in $\sigma$. This concludes the proof.
\end{proof}

\begin{theorem}\label{th:monotone-free-ordering-fixed-sides}
{\sc MNTLP with Fixed Side} is \NPC.
\end{theorem}

\noindent
\begin{proof}
The reduction presented in the proof of Theorem~\ref{th:monotone-free-ordering-free-sides}, equipped with the side assignment for the inter-cluster edges described in the direction ($\Longrightarrow$) implies the statement.
\end{proof}



Since the instances of {\sc MNTLP} used in the proof of Theorem~\ref{th:monotone-free-ordering-free-sides} are planar whenever they are locally planar, testing the existence of a monotone planar \nt representation with fixed square assignment is also \NPC. 
Further, the instances of {\sc \nt Planarity} used in the proof of Theorem~\ref{th:free-ordering-free-sides} can be drawn planarly with straight-line (hence monotone) edges, whenever they are planar. Hence, testing whether a flat clustered graph admits a monotone planar \nt representation -- without square assignment -- is also \NPC.

Consider now a flat clustered graph $G=(V,E,\mathcal{C})$ and a monotone \nt representation $\Gamma$ of $G$ with fixed square assignment $\gamma$. Consider two clusters $V_a,V_b\in \mathcal{C}$ and let $Q_a=\gamma(V_a)$ and $Q_b=\gamma(V_b)$. Since $Q_a$ and $Q_b$ are disjoint, there exists either a vertical or a horizontal line separating them. Suppose that the former holds, the other case being analogous. Also suppose that $\max_x(Q_a) < \min_x(Q_b)$ and $\max_y(Q_a) \geq \max_y(Q_b)$, the other cases being analogous up to reflections of the Cartesian axes (refer to
Fig.~\ref{fi:forbidden-edge-drawings}). Also, consider an inter-cluster edge $e = (u,v) \in E_{a,b}$. 
Depending on the relative positions of $Q_a$ and $Q_b$ in $\Gamma$, not all the possible combinations of side assignments for $e$ might be allowed, as described in the following property. Notice that, by the assumptions on the relative positions of $Q_a$ and $Q_b$ in $\Gamma$ and by the monotonicity and the local planarity of $\Gamma$, we have that $s_a(e)\neq \textrm{\sc l}, \textrm{\sc t}$ and $s_a(e)\neq \textrm{\sc r}$.

\begin{property}\label{prop:arrangements} 
Let $y_u$ and $y_v$ be the $y$-coordinate of points $m_\textrm{\sc r}^u$ and $m_\textrm{\sc l}^v$, respectively. The following three arrangements are possible for $Q_a$ and $Q_b$ in $\Gamma$.

{\em Arrangement~1:} $\max_y(Q_b) < \min_y(Q_a)$. Then $s_b(e)\neq \textrm{\sc b}$ and all other four side assignments $\langle s_a(e)=\textrm{\sc r}, s_b(e)=\textrm{\sc t} \rangle$, $\langle s_a(e)=\textrm{\sc r}, s_b(e)=\textrm{\sc l} \rangle$, $\langle s_a(e)=\textrm{\sc b}, s_b(e)=\textrm{\sc t} \rangle$, and 
$\langle s_a(e)=\textrm{\sc b}, s_b(e)=\textrm{\sc l} \rangle$ are allowed for $e$.

{\em Arrangement~2:}  $\min_y(Q_b)< \min_y(Q_a)\leq \max_y(Q_b)$. Then $s_b(e)\neq \textrm{\sc b}$; also, pair $\langle s_a(e)=\textrm{\sc b}, s_b(e)=\textrm{\sc t} \rangle$ is not allowed, while pair $\langle s_a(e)=\textrm{\sc r}, s_b(e)=\textrm{\sc l} \rangle$ is allowed. The remaining two possible pairs $\langle s_a(e)=\textrm{\sc r}, s_b(e)=\textrm{\sc t} \rangle$ and $\langle s_a(e)=\textrm{\sc b}, s_b(e)=\textrm{\sc l} \rangle$ are or are not allowed, depending on $y_u$ and $y_v$. In particular, if $y_u \leq \max_y(Q_b)$, then $\langle s_a(e)=\textrm{\sc r}, s_b(e)=\textrm{\sc t} \rangle$ is not allowed, otherwise it is; also, if $y_v \geq \min_y(Q_a)$, then $\langle
s_a(e)=\textrm{\sc b}, s_b(e)=\textrm{\sc l} \rangle$ is not allowed, otherwise it is.

{\em Arrangement~3:}  $\min_y(Q_a)\leq \min_y(Q_b)$. Then $s_a(e)\neq \textrm{\sc b}$; also, pair $\langle s_a(e)=\textrm{\sc r}, s_b(e)=\textrm{\sc l} \rangle$ is allowed. The remaining two possible pairs $\langle s_a(e)=\textrm{\sc r}, s_b(e)=\textrm{\sc t} \rangle$ and $\langle s_a(e)=\textrm{\sc r}, s_b(e)=\textrm{\sc b} \rangle$ are or are not allowed, depending on $y_u$. In particular, if $y_u \leq \max_y(Q_b)$, then $\langle s_a(e)=\textrm{\sc r}, s_b(e)=\textrm{\sc t} \rangle$ is not allowed, otherwise it is, and if $y_u \geq \min_y(Q_b)$, then $\langle s_a(e)=\textrm{\sc r}, s_b(e)=\textrm{\sc b} \rangle$ is not allowed, otherwise it is. 
\end{property}

\begin{figure}[tb!]
    \centering

    \subfloat[{Arrangement~1}]{
      \includegraphics[width=0.2\textwidth,page=1]{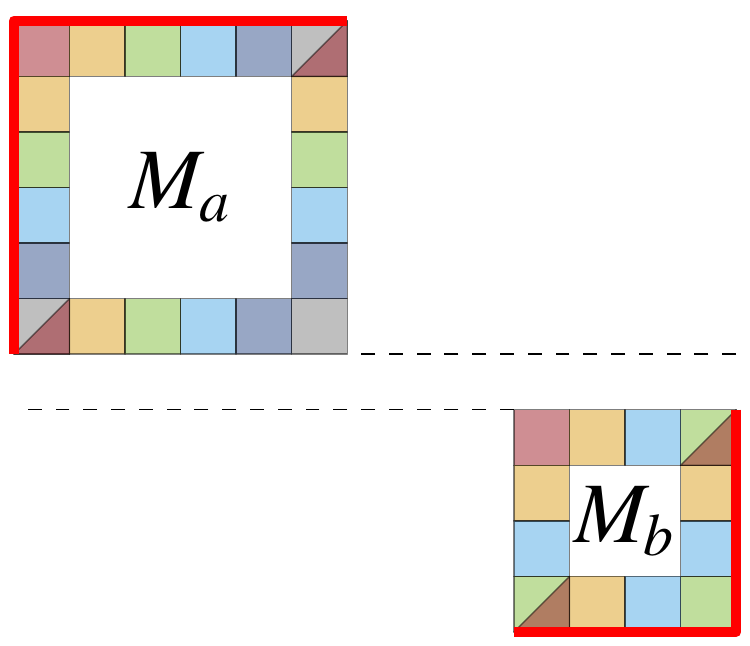}
    }
    \hfil
    \subfloat[{Arrangement~2}]{
      \includegraphics[width=0.2\textwidth,page=2]{img/forbidden-edge-drawings}
    }
    \hfil
    \subfloat[{Arrangement~3}]{
      \includegraphics[width=0.2\textwidth,page=3]{img/forbidden-edge-drawings}
    }
  \caption{Possible arrangements for squares $Q_a$ and $Q_b$. Thick red segments represent sides of $Q_a$ and $Q_b$ edge $(u,v)$ cannot be assigned to. Red curves show further forbidden side assignment pairs for edge $(u,v)$.}
  \label{fi:forbidden-edge-drawings}
  \end{figure}

Note that if an edge $e$ can be drawn as an $xy$-monotone curve not crossing any matrix, then it can also be drawn as a straight-line segment not crossing any matrix, since the pipe of $Q_a$ and $Q_b$ does not intersect any matrix other than $M_a$ and $M_b$. 
The next lemma extends this observation by arguing that the $xy$-monotonicity constraint can be replaced by a straight-line requirement 
also for what concerns crossings between inter-cluster edges incident to the same matrix.

\begin{lemma}\label{le:arrangements-straight-line} 
An instance $\langle G=(V,E,\mathcal{C}),\gamma \rangle$ of {\sc MNTLP with Fixed Order and Fixed Side} is locally planar if and only if it admits a monotone locally planar \nt representation in which all the inter-cluster edges are drawn as straight-line segments.
\end{lemma}


\noindent
\begin{proof}
Since a straight-line segment is an $xy$-monotone curve, one direction of the proof is trivial. Consider an \nt representation $\Gamma$ of $\langle G=(V,E,\mathcal{C}),\gamma \rangle$ with a fixed row-column order, a fixed side assignment, and a fixed square assignment, in which all the inter-cluster edges are  straight-line segments. Suppose that $\Gamma$ is not locally planar and consider two crossing inter-cluster edges $e=(v_{a,1},v_b)$ and $f=(v_{a,2},v_c)$ such that $v_{a,1}$ and $v_{a,2}$ belong to the same cluster $V_a\in \cal C$. We show that $e$ and $f$ cross in any monotone \nt representation $\Gamma'$ with the same row-column order, side assignment, and square assignment as $\Gamma$. Two are the cases: either $v_b$ and $v_c$ belong to the same cluster $V_b$ as in Fig.~\ref{fi:straightline1a} (Case 1), or they belong to different clusters $V_b$ and $V_c$, respectively, as in Fig.~\ref{fi:straightline1b} (Case 2).

  \begin{figure}[tb!]
    \centering
    \subfloat[]{
      \includegraphics[width=0.35\textwidth,page=1]{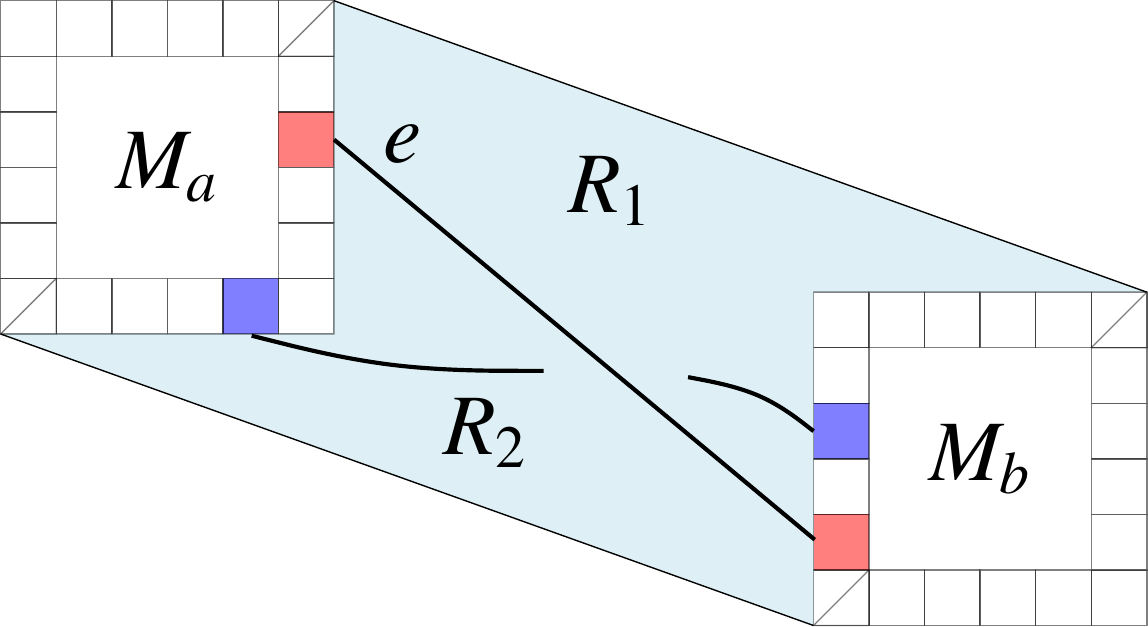}\label{fi:straightline1a}}
    \hfil
    \subfloat[]{
      \includegraphics[width=0.35\textwidth,page=2]{img/straightline1}\label{fi:straightline1b}}
  \caption{Illustration for the proof of Lemma~\ref{le:arrangements-straight-line}. (a) Case~1. (b) Case~2.}
  \label{fi:straightline}
  \end{figure}

In Case 1, consider the region $R_{ab}$ of the plane inside the pipe of $Q_a$ and $Q_b$ and outside each of $Q_a$ and $Q_b$. Edge $e$ splits $R_{ab}$ into two regions $R_1$ and $R_2$. Since $e$ and $f$ cross in $\Gamma$, the end-points of $f$ are one incident to $R_1$ and one incident to $R_2$ in $\Gamma'$. Since the representation of $f$ has to lie inside $R_{ab}$, it follows that $e$ and $f$ cross in $\Gamma'$.



In Case 2, consider the region $R_{ab}$ defined as in Case~1 and consider the region $R_{ac}$ of the plane inside the pipe of $Q_a$ and $Q_c$ and outside each of $Q_a$ and $Q_c$. Since $e$ lies in $R_{ab}$ and $f$ in $R_{ac}$, and since $e$ and $f$ cross in $\Gamma$, it follows that the intersection of $R_{ab}$ and $R_{ac}$ is a non-empty region $R_{\cap}$. The part of $e$ inside $R_{\cap}$ partitions $R_{\cap}$ into two regions $R_1$ and $R_2$. Since $e$ and $f$ cross in $\Gamma$, the end-point of $f$ on the boundary of $Q_a$ is incident to the one between $R_1$ and $R_2$ that does not share the boundary with region $R_{ac}-R_{\cap}$. Since the representation of $f$ has to lie inside $R_{ac}$, it follows that $e$ and $f$ cross in $\Gamma'$.
\end{proof}

The previous lemma, in contrast to the negative results of Theorems~\ref{th:monotone-free-ordering-free-sides} and~\ref{th:monotone-free-ordering-fixed-sides}, allows us to show that {\sc MNTLP with Fixed Order and Fixed Side} is a polynomial-time solvable problem.

\begin{theorem} \label{th:monotone-fixed-order-and-side-polynomial}
{\sc MNTLP with Fixed Order and Fixed Side} can be solved in polynomial time.
\end{theorem}

\begin{proof}
We check whether every edge can be represented as an $xy$-monotone curve by Property~\ref{prop:arrangements}. Further, we check whether all the pairs of inter-cluster edges incident to the same cluster admit a non-crossing straight-line drawing; by Lemma~\ref{le:arrangements-straight-line} this is equivalent to test the local planarity with fixed row-column order, fixed side assignment, and fixed square assignment of the given instance. 
\end{proof}


The remaining piece of the complexity puzzle for {\sc MNTLP} is the setting with fixed row-column order and free side assignment. Although we are not able to establish the complexity of the corresponding decision problem, we show that testing 
{\sc MNTLP} with fixed order is a polynomial-time solvable problem if the number of clusters is constant. In order to do that, we show how to transform the instances of our problem into instances of $2$-SAT. 

Assuming the hypotheses stated before Property~\ref{prop:arrangements} about the relative positions of $Q_a$ and $Q_b$, we say that an inter-cluster edge $e = (u \in V_a, v \in V_b)$ is {\em S-drawn} in $\Gamma$ if: 
\begin{itemize}
\item[$(i)$] $Q_a$ and $Q_b$ are arranged as in Arrangement~1 of Property~\ref{prop:arrangements} and either $\langle s_a(e)=\textrm{\sc r}, s_b(e)=\textrm{\sc l} \rangle$ or $\langle s_a(e)=\textrm{\sc b}, s_b(e)=\textrm{\sc t} \rangle$; or
\item[$(ii)$] $Q_a$ and $Q_b$ are arranged as in Arrangement~2 of Property~\ref{prop:arrangements} and it holds that (a) $\langle s_a(e)=\textrm{\sc r}, s_b(e)=\textrm{\sc l} \rangle$, (b) $y_u>\max_y(Q_b)$, and (c) $y_v<\min_y(Q_a)$.
\end{itemize}
Note that if $Q_a$ and $Q_b$ are arranged as in Arrangement~3 of Property~\ref{prop:arrangements}, then $e$ is not S-drawn in $\Gamma$, by definition. The representation of an S-drawn edge is an {\em S-drawing}. We have the following.


\begin{lemma}\label{le:two-drawings-with-S-FORMULA}
Let $\langle G=(V,E,\mathcal{C}=\{V_a,V_b\}), \gamma, \sigma \rangle$ be an
instance of {\sc MNTLP with Fixed Order}.
Consider the following two cases:

\begin{itemize}
\item Case~1: an inter-cluster edge $e^* \in E$ has a given S-drawing $\Gamma_e$, or
\item Case~2: no inter-cluster edge in $E$ has an S-drawing.
\end{itemize}

Both in Case~1 and in Case~2, we can construct in $O(|E|^2)$ time a $2$-SAT formula $\phi(a,b,\Gamma_e)$ and $\phi(a,b)$, respectively, with length $O(|E|^2)$ that is satisfiable if and only if $\langle
G, \gamma, \sigma \rangle$ admits a monotone locally planar \nt representation with fixed order satisfying the constraint of the corresponding~case.
\end{lemma}

\noindent
\begin{proof}
Consider the squares $Q_a=\gamma(V_a)$ and $Q_b=\gamma(V_b)$. If they are not disjoint, no \nt representation of $G$ exists, hence the statement is trivially true. Otherwise, there exists either a vertical line or a horizontal line separating them. Suppose that the former holds, the other case being analogous. Suppose that $\max_x(Q_a) < \min_x(Q_b)$ and $\max_y(Q_a) \geq \max_y(Q_b)$, the
other cases being analogous up to reflections of the Cartesian axes.

%
%



Suppose that an inter-cluster edge $e^*$ is required to have a drawing $\Gamma_e$ as in Case~1. By the definition of an S-drawn edge, if $Q_a$ and $Q_b$ are arranged as in Arrangement~3 of Property~\ref{prop:arrangements}, then the required \nt representation does not exist, thus the statement trivially holds. Hence, we can assume that $Q_a$ and $Q_b$ are arranged as in Arrangement~1 or~2 of Property~\ref{prop:arrangements}. Let $e \neq e^* \in E$ be any inter-cluster edge not adjacent to $e$. Denote by $\sigma_a$ and $\sigma_b$ the row-column orders associated to $V_a$ and $V_b$ in $\sigma$, respectively.


Consider Arrangement~1 and suppose $s_a(e^*)=\textrm{\sc r}$ and $s_b(e^*)=\textrm{\sc l}$. The end-vertices of $e$ and $e^*$ in $V_a$ (in $V_b$) have two possible relative positions in $\sigma_a$ (resp.\ in $\sigma_b$). This leads to four possible combinations for these relative positions. 

If $\sigma_a(e^*) < \sigma_a(e)$ and $\sigma_b(e) < \sigma_b(e^*)$, then any $xy$-monotone curve representing $e$ crosses $e^*$, independently of the side assignment for $e$, and the statement trivially holds.  See Fig.~\ref{fi:lemma-s-case-1-arrangement-1-d}.  For each of the three remaining combinations, {\em exactly two} side assignments for $e$ create no crossing with $e^*$. Indeed:

\begin{itemize}
\item If $\sigma_a(e) < \sigma_a(e^*)$ and $\sigma_b(e) < \sigma_b(e^*)$, then it holds true that either $s_a(e)=\textrm{\sc r}$ and $s_b(e) = \textrm{\sc t}$, or that $s_a(e)=\textrm{\sc r}$
and $s_b(e) = \textrm{\sc l}$. See Fig.~\ref{fi:lemma-s-case-1-arrangement-1-a}.
\item If $\sigma_a(e) < \sigma_a(e^*)$ and $\sigma_b(e^*) < \sigma_b(e)$, then it holds true that either $s_a(e)=\textrm{\sc r}$ and $s_b(e) = \textrm{\sc t}$, or that $s_a(e)=\textrm{\sc b}$
and $s_b(e) = \textrm{\sc l}$. See Fig.~\ref{fi:lemma-s-case-1-arrangement-1-b}.
\item If $\sigma_a(e^*) < \sigma_a(e)$ and $\sigma_b(e^*) < \sigma_b(e)$, then it holds true that either $s_a(e)=\textrm{\sc r}$ and $s_b(e) = \textrm{\sc l}$, or that $s_a(e)=\textrm{\sc b}$
and $s_b(e) = \textrm{\sc l}$. See Fig.~\ref{fi:lemma-s-case-1-arrangement-1-c}.
\end{itemize}

  \begin{figure}[tb!]
    \centering
    \subfloat[]{
      \includegraphics[width=0.2\textwidth,page=4]{img/lemma-s-case-1-arrangement-1}\label{fi:lemma-s-case-1-arrangement-1-d}
    }
    \hfil
    \subfloat[]{
      \includegraphics[width=0.2\textwidth,page=1]{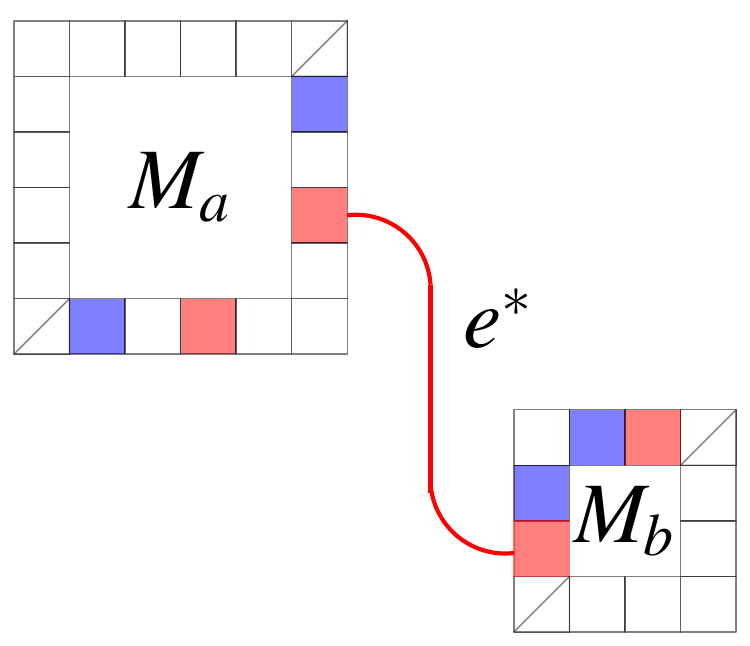}\label{fi:lemma-s-case-1-arrangement-1-a}
    }
    \hfil
    \subfloat[]{
      \includegraphics[width=0.2\textwidth,page=2]{img/lemma-s-case-1-arrangement-1}\label{fi:lemma-s-case-1-arrangement-1-b}
    }
    \hfil
    \subfloat[]{
      \includegraphics[width=0.2\textwidth,page=3]{img/lemma-s-case-1-arrangement-1}\label{fi:lemma-s-case-1-arrangement-1-c}
    }
     
  \caption{
  Illustrations for the proof of Lemma~\ref{le:two-drawings-with-S-FORMULA}, Case 1, Arrangement 1.
  }
  \label{fi:lemma-s-case-1-arrangement-1}
  \end{figure}

The discussion for the case in which $Q_a$ and $Q_b$ are arranged as in Arrangement~1, $s_a(e^*)=\textrm{\sc b}$, and $s_b(e^*)=\textrm{\sc t}$ is analogous to the previous one.

Consider now Arrangement~2. According to the definition of S-drawing it holds true
for $e^*=(u,v)$ that (a) $s_a(e^*)=\textrm{\sc r}$ and $s_b(e^*)=\textrm{\sc l}$, (b) the $y$-coordinate of $p_u$ is greater than $\max_y(Q_b)$, and
(c) the $y$-coordinate of $p_v$ is smaller than $\min_y(Q_a)$. 

Similarly to Arrangement~1, there are four possible combinations for the relative positions of the end-vertices of $e$ and $e^*$ in $\sigma_a$ and $\sigma_b$. If $\sigma_a(e^*) < \sigma_a(e)$ and $\sigma_b(e) < \sigma_b(e^*)$, then any $xy$-monotone curve representing $e$ crosses $e^*$, independently of the side assignment for $e$, and the statement trivially holds.  See Fig.~\ref{fi:lemma-s-case-1-arrangement-2-d}.  For each of the three remaining combinations, {\em exactly two} side assignments for $e$ create no crossing with $e^*$. 

\begin{itemize}
\item If $\sigma_a(e) < \sigma_a(e^*)$ and $\sigma_b(e) < \sigma_b(e^*)$, then it holds true that either $s_a(e)=\textrm{\sc r}$ and $s_b(e) = \textrm{\sc t}$, or that $s_a(e)=\textrm{\sc r}$
and $s_b(e) = \textrm{\sc l}$. See Fig.~\ref{fi:lemma-s-case-1-arrangement-2-a}.
\item If $\sigma_a(e) < \sigma_a(e^*)$ and $\sigma_b(e^*) < \sigma_b(e)$, then it holds true that either $s_a(e)=\textrm{\sc r}$ and $s_b(e) = \textrm{\sc t}$, or that $s_a(e)=\textrm{\sc b}$
and $s_b(e) = \textrm{\sc l}$. See Fig.~\ref{fi:lemma-s-case-1-arrangement-2-b}.
\item If $\sigma_a(e^*) < \sigma_a(e)$ and $\sigma_b(e^*) < \sigma_b(e)$, then it holds true that either $s_a(e)=\textrm{\sc r}$ and $s_b(e) = \textrm{\sc l}$, or that $s_a(e)=\textrm{\sc b}$
and $s_b(e) = \textrm{\sc l}$. 
See Fig.~\ref{fi:lemma-s-case-1-arrangement-2-c}.
\end{itemize}

  \begin{figure}[tb!]
    \centering
    \subfloat[]{
      \includegraphics[width=0.2\textwidth,page=4]{img/lemma-s-case-1-arrangement-2}\label{fi:lemma-s-case-1-arrangement-2-d}
    }
    \hfil
    \subfloat[]{
      \includegraphics[width=0.2\textwidth,page=1]{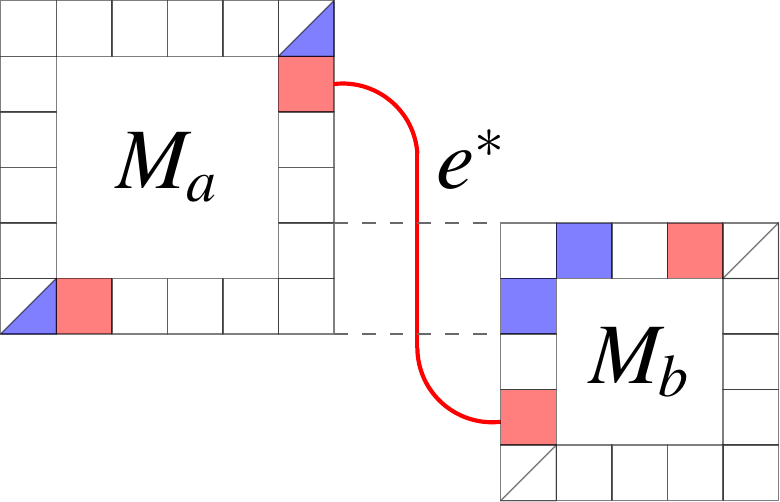}\label{fi:lemma-s-case-1-arrangement-2-a}
    }
    \hfil
    \subfloat[]{
      \includegraphics[width=0.2\textwidth,page=2]{img/lemma-s-case-1-arrangement-2}\label{fi:lemma-s-case-1-arrangement-2-b}
    }
    \hfil
    \subfloat[]{
      \includegraphics[width=0.2\textwidth,page=3]{img/lemma-s-case-1-arrangement-2}\label{fi:lemma-s-case-1-arrangement-2-c}
    }
  \caption{
  Illustrations for the proof of Lemma~\ref{le:two-drawings-with-S-FORMULA}, Case 1, Arrangement 2.
  }
  \label{fi:lemma-s-case-1-arrangement-2}
  \end{figure}

Hence, for each inter-cluster edge $e \neq e^* \in E$ not adjacent to $e^*$, there exist two side assignments for $e$ that allow it to be represented as an $xy$-monotone curve not intersecting $e^*$. 

We are now ready to show, for Case~1 of the lemma, that a monotone locally planar \nt representation of $\langle G=(V,E,\mathcal{C}=\{V_a,V_b\}), \gamma \rangle$ in which $e^*$ is represented by $\Gamma_e$ exists if and only if a suitable $2$-SAT formula $\phi(a,b,\Gamma_e)$ is satisfiable.

For each inter-cluster edge $e \neq e^* \in E$ not adjacent to $e^*$, we define a Boolean variable $x_e$. 
The above discussion shows that, if we did not conclude that a trivially false formula exists, then there are exactly two distinct side assignments for $e$. We select one arbitrarily, which we call {\em canonical side assignment}, and associate $x_e={\textrm{\sc{true}}}$ to it and $x_e={\textrm{\sc{false}}}$
to the other.



For each pair of non-adjacent inter-cluster edges $e_1,e_2 \neq e^* \in E$, consider the four possible side assignments for them. We add to
$\phi(a,b,\Gamma_e)$ at most four clauses defined as follows.
\begin{itemize}
\item If the canonical side assignment for $e_1$ and the canonical side
assignment for $e_2$ generate a crossing between $e_1$ and $e_2$, then we add
clause $\{\overline{x_{e_1}} \vee \overline{x_{e_2}}\}$ to $\phi(a,b,\Gamma_e)$.
\item If the canonical side assignment for $e_1$ and the non-canonical side
assignment for $e_2$ generate a crossing between $e_1$ and $e_2$, then
we add clause $\{\overline{x_{e_1}} \vee x_{e_2}\}$ to $\phi(a,b,\Gamma_e)$.
\item If the non-canonical side assignment for $e_1$ and the canonical
side assignment for $e_2$ generate a crossing between $e_1$ and $e_2$, then we
add clause $\{x_{e_1} \vee \overline{x_{e_2}}\}$ to $\phi(a,b,\Gamma_e)$.
\item If the non-canonical side assignment for $e_1$ and the non-canonical side assignment for $e_2$ generate a crossing between $e_1$
and $e_2$, then we add clause $\{x_{e_1} \vee x_{e_2}\}$ to $\phi(a,b,\Gamma_e)$.
\end{itemize}

As a consequence of the above discussion $\langle G=(V,E,\mathcal{C}=\{V_a,V_b\}), \gamma \rangle$
admits a monotone locally planar \nt representation in which $e^*$ is represented by $\Gamma_e$ if and only if $\phi(a,b,\Gamma_e)$ is satisfiable. Further, since the
number of clauses in $\phi(a,b,\Gamma_e)$ is upper-bounded by $O(|E|^2)$ and since it
can be determined in constant time whether a side assignment for any two edges
produces a crossing, then formula $\phi(a,b,\Gamma_e)$ can be constructed in $O(|E|^2)$
time and has $O(|E|^2)$ size. Since $2$-SAT formulae can be tested for
satisfiability in linear time~\cite{apt-lattcqbf-79}, the statement of
Case~1 follows.


Suppose now that Case~2 of the statement holds. According to  Property~\ref{prop:arrangements}, squares $Q_a$ and $Q_b$ can be arranged as in Arrangement~1,~2, or~3.

Consider Arrangement~1. By the hypothesis of the case, no edge is allowed to be S-drawn. Hence, for each inter-cluster edge $e$, we have either $s_a(e)=\textrm{\sc r}$ and $s_b(e)=\textrm{\sc t}$ or $s_a(e)=\textrm{\sc b}$ and $s_b(e)=\textrm{\sc l}$.

  \begin{figure}[tb!]
    \centering
    \subfloat[]{
      \includegraphics[width=0.2\textwidth,page=1]{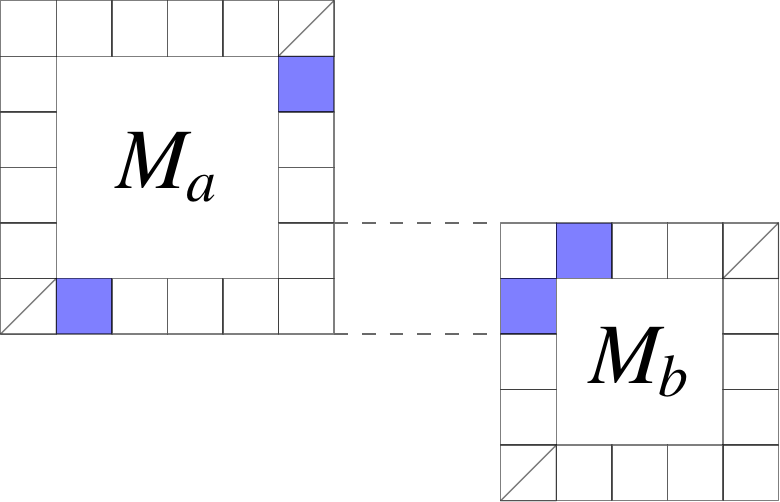}\label{fi:lemma-s-case-2-arrangement-2-a}
    }
    \hfil
    \subfloat[]{
      \includegraphics[width=0.2\textwidth,page=2]{img/lemma-s-case-2-arrangement-2}\label{fi:lemma-s-case-2-arrangement-2-b}
    }
    \hfil
    \subfloat[]{
      \includegraphics[width=0.2\textwidth,page=3]{img/lemma-s-case-2-arrangement-2}\label{fi:lemma-s-case-2-arrangement-2-c}
    }
    \hfil
    \subfloat[]{
      \includegraphics[width=0.2\textwidth,page=4]{img/lemma-s-case-2-arrangement-2}\label{fi:lemma-s-case-2-arrangement-2-d}
    }
  \caption{  Illustrations for the proof of Lemma~\ref{le:two-drawings-with-S-FORMULA}, Case 2, Arrangement 2.
  }
  \label{fi:lemma-s-case-2-arrangement-2}
  \end{figure}

Consider Arrangement~2. Let $e=(u,v)$ be an inter-cluster edge. We distinguish four cases depending on the $y$-coordinate $y_u$ of $m_\textrm{\sc r}^u$ with respect to $\max_y(Q_b)$ and on the $y$-coordinate $y_v$ of $m_\textrm{\sc l}^v$ with respect to $\min_y(Q_a)$. In each of the four cases, {\em at most two} side assignments for $e$ are possible so that $e$ is not S-drawn.
\begin{itemize}
\item If $y_u>\max_y(Q_b)$ and $y_v\geq \min_y(Q_a)$, then it holds true that either
$s_a(e)=\textrm{\sc r}$ and $s_b(e)=\textrm{\sc t}$, or that $s_a(e)=\textrm{\sc r}$ and $s_b(e)=\textrm{\sc l}$. See Fig.~\ref{fi:lemma-s-case-2-arrangement-2-a}.
\item If $y_u>\max_y(Q_b)$ and $y_v<\min_y(Q_a)$, then it holds true that either
$s_a(e)=\textrm{\sc r}$ and $s_b(e)=\textrm{\sc t}$, or that $s_a(e)=\textrm{\sc b}$ and $s_b(e)=\textrm{\sc l}$. See Fig.~\ref{fi:lemma-s-case-2-arrangement-2-b}; notice that the side assignment $s_a(e)=\textrm{\sc r}$ and $s_b(e)=\textrm{\sc l}$ would imply that $e$ is S-drawn, which is not possible by hypothesis.
\item If $y_u\leq\max_y(Q_b)$ and $y_v\geq \min_y(Q_a)$, then it holds true that $s_a(e)=\textrm{\sc r}$ and $s_b(e)=\textrm{\sc l}$. See Fig.~\ref{fi:lemma-s-case-2-arrangement-2-c}.
\item If $y_u\leq \max_y(Q_b)$ and $y_v<\min_y(Q_a)$, then it holds true that either
$s_a(e)=\textrm{\sc r}$ and $s_b(e)=\textrm{\sc l}$, or that $s_a(e)=\textrm{\sc b}$ and $s_b(e)=\textrm{\sc l}$. See Fig.~\ref{fi:lemma-s-case-2-arrangement-2-d}.
\end{itemize}

Consider Arrangement~3. Let $e=(u,v)$ be an inter-cluster edge. By definition
$e$ cannot be S-drawn. We distinguish three cases depending on the $y$-coordinate $y_u$ of $m_\textrm{\sc r}^u$ with respect to $\min_y(Q_b)$ and $\max_y(Q_b)$. In each of the three cases, {\em at most two} side assignments for $e$ are possible.

\begin{itemize}
\item If $y_u>\max_y(Q_b)$, then it holds true that either $s_a(e)=\textrm{\sc r}$ and $s_b(e)=\textrm{\sc t}$, or that $s_a(e)=\textrm{\sc r}$ and $s_b(e)=\textrm{\sc l}$. 
See Fig.~\ref{fi:lemma-s-case-2-arrangement-3-a}.
\item If $\min_y(Q_b)\leq y_u\leq \max_y(Q_b)$, then it holds true that $s_a(e)=\textrm{\sc r}$ and $s_b(e)=\textrm{\sc l}$. See Fig.~\ref{fi:lemma-s-case-2-arrangement-3-b}.
\item If $y_u<\min_y(Q_b)$, then it holds true that either $s_a(e)=\textrm{\sc r}$ and $s_b(e)=\textrm{\sc l}$, or that $s_a(e)=\textrm{\sc r}$ and $s_b(e)=\textrm{\sc b}$. See Fig.~\ref{fi:lemma-s-case-2-arrangement-3-c}.
\end{itemize}

  \begin{figure}[tb!]
    \centering
    \subfloat[]{
      \includegraphics[width=0.2\textwidth,page=1]{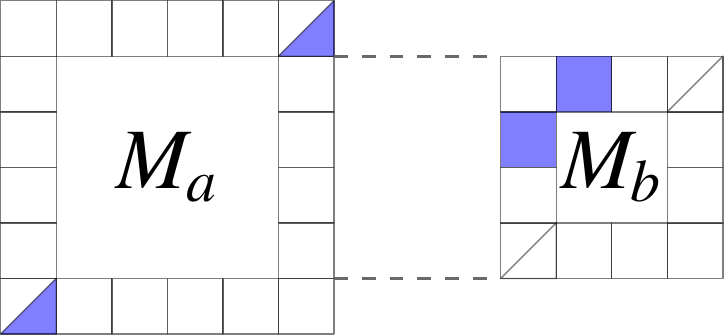}\label{fi:lemma-s-case-2-arrangement-3-a}
    }
    \hfil
    \subfloat[]{
      \includegraphics[width=0.2\textwidth,page=2]{img/lemma-s-case-2-arrangement-3}\label{fi:lemma-s-case-2-arrangement-3-b}
    }
    \hfil
    \subfloat[]{
      \includegraphics[width=0.2\textwidth,page=3]{img/lemma-s-case-2-arrangement-3}\label{fi:lemma-s-case-2-arrangement-3-c}
    }
     
  \caption{
  Illustrations for the proof of Lemma~\ref{le:two-drawings-with-S-FORMULA}, Case 2, Arrangement 3.
  }
  \label{fi:lemma-s-case-2-arrangement-3}
  \end{figure}

Hence, regardless of whether $Q_a$ and $Q_b$ are arranged as in Arrangement~1,~2, or~3, and regardless of the $y$-coordinate of $m_\textrm{\sc r}^u$ and $m_\textrm{\sc l}^v$, there exist at most two side assignments for $e$ that allow it to be represented as an $xy$-monotone curve. 


The construction of the $2$-SAT formula and the bound on its size can be derived analogously to Case~1; the only difference is that, when only one side assignment is possible, a clause with a single literal is generated. This concludes the proof of the lemma.
\end{proof}

We now turn to the study of flat clustered graphs with three clusters.

\begin{lemma}\label{le:two-drawings-three-matrices-FORMULA}
Let $\langle G=(V,E,\mathcal{C}=\{V_a,V_b,V_c\}), \gamma, \sigma \rangle$ be an
instance of {\sc MNTLP with Fixed Order}. 
Consider the four cases that are generated by assuming that an edge $e^* \in E_{a,b}$ has a prescribed S-drawing or not and that an edge $f^* \in E_{a,c}$ has a prescribed S-drawing or not.
In each case, we can construct in $O(|E|^2)$ time a $2$-SAT formula $\phi(a,b,c)$ with
length $O(|E|^2)$ that is satisfiable if and only if $\langle
G, \gamma, \sigma \rangle$ admits a monotone \nt representation with fixed order that satisfies the constraints of the corresponding case, such that no inter-cluster edge intersects any matrix except at its incidence points, and such that there are no two edges, one in $E_{a,b}$ and one in $E_{a,c}$, that cross each other.
\end{lemma}

\begin{proof}
In each of the four cases, the hypotheses lead us in either Case~1 or Case~2 of Lemma~\ref{le:two-drawings-with-S-FORMULA} for the edges in $E_{a,b}$ and the same holds for the edges in $E_{a,c}$. Hence, by Lemma~\ref{le:two-drawings-with-S-FORMULA}, each of these edges admits at most two side assignments in each case. Moreover, each of these side assignments corresponds to a directed or negated literal. For each pair of edges $e\in E_{a,b}$ and $f\in E_{a,c}$ and for each of the at most four side assignments for them, we exploit Lemma~\ref{le:arrangements-straight-line} to test whether a side assignment for $e$ and $f$ leads to a crossing and in the case of a crossing we introduce suitable clauses to rule out that side assignment.  
\end{proof}

We finally get the following.

\begin{theorem}\label{th:monotone-fixed-ordering-free-sides}
{\sc MNTLP with Fixed Ordering}
can be tested in 
$|E|^{O({|\mathcal{C}|^2})}$
time for an instance $\langle G=(V,E,\mathcal{C}), \gamma, \sigma \rangle$.
\end{theorem}

\noindent
\begin{proof}
For each pair $V_a,V_b$ of adjacent clusters in $\cal C$, we guess whether
$V_a,V_b$ belongs to a set $\mathcal P_{\textrm{\sc s}}$ or to a set $\mathcal
P_{\textrm{{\sc n}}}$. The set $\mathcal P_{\textrm{\sc s}}$ contains all the pairs $V_a,V_b$ of clusters that have an inter-cluster edge that is S-drawn. The set $\mathcal P_{\textrm{{\sc n}}}$ contains all the pairs $V_a,V_b$ of clusters that do not have an inter-cluster edge that is S-drawn. For each pair $V_a,V_b$ of clusters in $\mathcal P_{\textrm{\sc s}}$ we guess an inter-cluster edge $e \in E_{a,b}$ that can be S-drawn and one of its possible S-drawings $\Gamma_e$ for $e$; we remark that the guess of $\Gamma_e$ consists of a guess of the side assignement for $e$, hence there are a constant number of possible guesses for each edge $e$. 

By means of Lemma~\ref{le:two-drawings-with-S-FORMULA} we compute the following formula:

$$\phi_{\textrm{pairs}} = \bigwedge_{V_a,V_b \in {\mathcal P_{\textrm{\sc s}}}} \phi(a,b,\Gamma_e) \bigwedge_{V_a,V_b \in {\mathcal P_{\textrm{{\sc n}}}}} \phi(a,b).$$

Further, let $\mathcal P_{\textrm{triplet}}$ be the set of triplets $V_a,V_b,V_c$ of clusters in $\cal C$ such that $V_b$ and $V_c$ are adjacent to $V_a$. We write  one of the formulae $\phi(a,b,c)$ of the four cases of Lemma~\ref{le:two-drawings-three-matrices-FORMULA} according to the presence in $\mathcal P_{\textrm{\sc s}}$ of an inter-cluster edge between $V_a$ and $V_b$ or of an inter-cluster edge between $V_a$ and $V_c$. By means of Lemma~\ref{le:two-drawings-three-matrices-FORMULA}, we compute the following:

$$\phi_{\textrm{triplets}} = \bigwedge_{V_a,V_b,V_c \in {\mathcal P_{\textrm{triplet}}}}
\phi(a,b,c).$$

Finally, we define $$\phi = \phi_{\textrm{pairs}} \wedge \phi_{\textrm{triplets}}.$$

We have that instance $\langle G=(V,E,\mathcal{C}), \gamma, \sigma \rangle$  is
a positive instance if and only if there exists a guess such that the
corresponding formula $\phi$ is satisfiable.

About the time complexity, for each guess $O(|E|^2)$ time is needed to compute the corresponding formula $\phi$ and to check it for satisfiability, due to Lemmata~\ref{le:two-drawings-with-S-FORMULA} and~\ref{le:two-drawings-three-matrices-FORMULA}. The number of guesses can be bounded as follows. For each pair of adjacent clusters $V_a,V_b$ we have to guess among $2|E_{a,b}|+1$ possibilities, corresponding to the choice of $|E_{a,b}|$ edges to be S-drawn, each in two possible ways, plus the possibility of not having any S-drawn edge.  This leads to $O((2|E|+1)^{|\mathcal{C}|^2})$, which is in $|E|^{O({|\mathcal{C}|^2})}$, guesses.
\end{proof}

Observe that the computational complexity of the algorithm described in the proof of Theorem~\ref{th:monotone-fixed-ordering-free-sides} is polynomial if the number of clusters is constant.



\subsection{A JavaScript Library for Monotone NodeTrix Representations}\label{se:editor}

In this section we consider the following scenario. A user moves (e.g.\ via a drag-and-drop primitive) the matrices representing the clusters of a flat clustered graph, choosing also her preferred row-column order. A system automatically selects sides for the inter-cluster edges so to produce a monotone \nt representation $\Gamma$ with a ``small'' $\chi(\Gamma)$.

The algorithm in the proof of Theorem~\ref{th:monotone-fixed-ordering-free-sides} suggests the following strategy: 
\begin{enumerate}
\item Compute the $2$-SAT formula associated to each
possible guess of S-drawings of edges between adjacent matrices. In each formula the value of a variable represents the two possible side assignments for an edge; further, each unsatisfied clause corresponds to a crossing in the monotone \nt representation.  
\item If one of such formulae admits a solution, draw the edges according to the values of the associated variables, obtaining an MNTLP representation. 
\item Otherwise, for each formula, heuristically search for a solution of the corresponding MAX-$2$-SAT problem and keep the solution with the smallest number of false clauses, corresponding to a drawing with few local crossings.
\end{enumerate}
Such a strategy requires polynomial time if the number of clusters is constant (Theorem~\ref{th:monotone-fixed-ordering-free-sides}) and the selected MAX-$2$-SAT heuristic is polynomial. However, solving a MAX-$2$-SAT instance for each of the guesses of the proof of Theorem~\ref{th:monotone-fixed-ordering-free-sides} is unfeasible even in a static setting. 
%

Therefore, we modify the above strategy as follows. We restrict to monotone
\nt representations without S-drawn edges. A locally-planar flat clustered graph may become non-planar with this restriction, hence this choice corresponds to trading accuracy for efficiency. However, the proof of Theorem~\ref{th:monotone-fixed-ordering-free-sides} shows that in this setting there is a unique formula associated to an instance, hence we need to solve one MAX-$2$-SAT instance.

\begin{figure}[htb]
    \centering
	\includegraphics[width=0.9\textwidth]{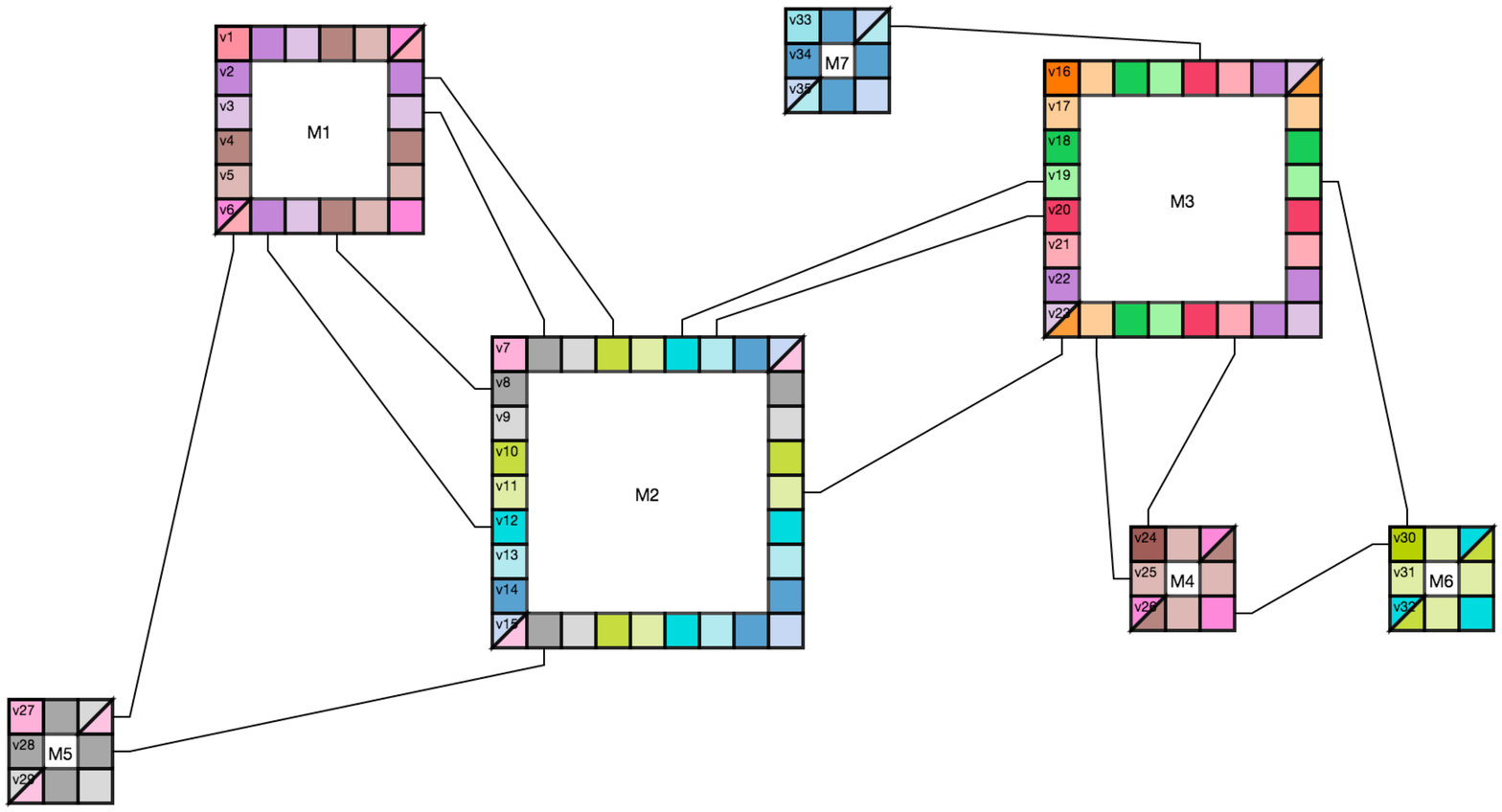}
    \caption{A NodeTrix Representation created by the demo editor~\cite{giordano-demo}.}\label{fi:overview}
\end{figure}

A JavaScript library implementing the above heuristic has been designed and used
in a proof-of-concept editor available at~\cite{giordano-demo} (see Fig.~\ref{fi:overview}). 
In the editor the internal part of the matrices is not shown and inter-cluster edges are polylines or splines; this is not intended to be the best choice and many alternatives for the actual geometry of the edges are possible whose visual appeal should be considered according to the specific application domain.
In order to check if the strategy is usable on medium size instances,
we experimented our simple editor on clustered graphs with at most twenty clusters
and 200 inter-cluster edges, without experiencing delays in the
interaction.
We did not compare the number of crossings produced by our heuristic with other approaches because, as far as we know, this is the first attempt to reduce local crossings in \nt representations. Also, it would be pointless to compare our approach with the original \nt software, since in that case the edges just attach to the nearest sides.
The JavaScript software of our library is freely available and can be integrated in any NodeTrix-style interface. 
As an example, it can be coupled with an algorithm that automatically places matrices based on a force-directed approach or with one that computes row-column order for the matrices with the purpose of clarifying the internal structure of the clusters~\cite{dbf-gb-05}.


\section{Conclusions and Open problems}\label{se:conclusions}

We have shown that clustered graphs for Nodetrix planarity is \NPC even if the order of the rows and columns is fixed or if the matrix sides to which the inter-cluster edges attach is fixed.
We have also studied the setting where matrices have fixed positions and inter-cluster edges are $xy$-monotone curves. In this case we established negative and positive results; leveraging on the latter, we developed a library that computes a layout of the inter-cluster edges with few crossings. A demo~\cite{giordano-demo} shows that the computation allows the user to move matrices without any slowdown of the interaction.

Several theoretical problems are related to the planarity of Nodetrix representations.
First, the \NPCN of Nodetrix planarity can be interpreted as a
proof of the \NPCN of clustered planarity (see, for example, \cite{addfpr-rccp-14,tibp-addfr-15,cdfpp-cccg-j-08,fce-pcg-95}) when a
specific type of representation is required. Observe, though, that a flat clustered 
graph may be Nodetrix planar even if its underlying graph is not planar.
Second, planarity of hybrid representations have been recently studied~\cite{addfp-ilrg-15} in the setting in which clusters are
represented as the intersections of geometric objects. Our results can be viewed
as a further progress in this area.
Third, given a flat clustered graph with two clusters, computing a locally planar Nodetrix representation in which the clusters are represented as matrices aligned along their principal diagonal is equivalent to solve the $2$-page bipartite book embedding with spine crossings problem~\cite{addfp-ilrg-15}. Interestingly, if the two matrices are aligned along their secondary diagonal this equivalence is not evident anymore.

Among the future research directions, we mention the one of automatically embedding the matrices in order to minimize the number of crossings in monotone Nodetrix representations. 

\bibliographystyle{splncs03}
\bibliography{bibliography}




\end{document}